\documentclass{article}

\usepackage{authblk}
\usepackage[utf8]{inputenc} 
\usepackage[T1]{fontenc}    
\usepackage{hyperref}       
\usepackage{url}            
\usepackage{booktabs}       
\usepackage{amsfonts}       
\usepackage{nicefrac}       
\usepackage{microtype}      
\usepackage{lipsum}		
\usepackage{graphicx}
\usepackage{doi}
\usepackage{amsmath, amsthm}
\allowdisplaybreaks

\usepackage{setspace}

\usepackage{multirow}
\usepackage{footnote}
\usepackage{float}
\usepackage{comment}
\usepackage{enumerate}
\usepackage{xcolor}
\usepackage{amssymb}
\usepackage{rotating}
\usepackage[numbers,super]{natbib}
\usepackage{hyperref}
\hypersetup{colorlinks, citecolor={blue!70!black}, linkcolor={blue!70!black}, urlcolor={blue!80!black}}
\usepackage{titlesec}

\usepackage{listings}

\usepackage[verbose=true,letterpaper]{geometry}
\AtBeginDocument{
  \newgeometry{
    textheight=9in,
    textwidth=6.5in,
    top=1in,
    headheight=14pt,
    headsep=25pt,
    footskip=30pt
  }
}

\usepackage{setspace}
\newcommand{\one}{\mathbf{1}}

\newcommand{\be}{\mathbf{e}}

\newcommand{\bz}{\mathbf{z}}

\newcommand{\bee}{\begin{eqnarray*}}
\newcommand{\eee}{\end{eqnarray*}}
\newcommand{\bi}{\begin{enumerate}}
\newcommand{\ei}{\end{enumerate}}

\newtheorem*{theorem*}{Theorem}

\newtheorem{theorem}{Theorem}

\newtheorem{lemma}{Lemma}

\def\spacingset#1{
  \renewcommand{\baselinestretch}{#1}\small\normalsize
  \setlength{\abovedisplayskip}{4pt}
  \setlength{\belowdisplayskip}{6pt}
  \setlength{\abovedisplayshortskip}{2pt}
  \setlength{\belowdisplayshortskip}{4pt}
}

\title{Propensity score weighted Cox regression for survival outcomes in observational studies with multiple or factorial treatments}

\author[1]{Zixian Zhao\textsuperscript{$\dagger$}}
\author[2]{Chengxin Yang\textsuperscript{$\dagger$}}
\author[3]{Fan Li\thanks{Corresponding author, Fan Li (email: fl35@duke.edu).}}
\affil[1]{Qiuzhen College, Tsinghua University, Beijing, China}
\affil[2]{Department of Biostatistics and Bioinformatics, Duke University, North Carolina, USA}
\affil[3]{Department of Statistical Science, Duke University, North Carolina, USA}

\begin{document}
\date{}

\maketitle

\begingroup
\renewcommand{\thefootnote}{\ensuremath{\dagger}}
\footnotetext[0]{Equal contribution.}
\endgroup

\begin{abstract}
In observational studies with survival or time-to-event outcomes, a propensity score weighted marginal Cox proportional hazards model with the treatment variable as the only predictor is commonly used to estimate the marginal hazard ratio between two treatments {as the causal effect}. Observational studies often have more than two treatments, but corresponding analysis methods are limited. In this paper, we combine propensity score weighting for multiple treatments and a marginal Cox model with indicators for each treatment to estimate the {marginal} hazard ratios between multiple treatments and a common reference treatment. We illustrate two weighting schemes: inverse probability of treatment weighting and overlap weighting. We prove the consistency of the maximum weighted partial likelihood estimator of the marginal hazard ratio and derive a robust sandwich variance estimator. As an important special case of multiple treatments, we elaborate on the Cox model for two-way factorial treatments. We apply the method to evaluate the real-world comparative effectiveness of three types of anti-obesity medications on heart failure. We develop an associated R package \texttt{PSsurvival}.
\end{abstract}

\textbf{Keywords: causal inference, factorial design, hazard ratio, multiple treatment, propensity score, survival analysis, target trial emulation} 

\newpage
\spacingset{1.5}

\section{Introduction}

The Cox proportional hazards regression \cite{cox1972regression} is the most widely used model for survival or general time-to-event outcomes in medical studies. In causal inference with observational data, the Cox model is often combined with propensity scores \cite{rosenbaum1983central} to adjust for confounding {to study the causal effect on a target population}. In particular, {the marginal hazard ratio (MHR)}, the coefficient of the treatment variable in a propensity weighted Cox model with the treatment variable as the only predictor, is reported as the {causal estimand}. \cite{hernan2000marginal, Austin2013estimate, austin2014use, austin2017performance} {Analysts can choose different weighting schemes, such as the inverse probability weighting or overlap weighting, to estimate the causal effect on the target population of interest.} 
Multiple treatments are common in clinical research. For example, our motivating application aims to evaluate the real-world comparative effectiveness of three anti-obesity medications on heart failure endpoints: glucagon-like peptide-1 (GLP-1) receptor agonists, naltrexone-bupropion, and phentermine-topiramate; the patient population of each medication differs substantially in baseline characteristics, and therefore, we must adjust for these differences in estimating the causal effects.  {However, most existing work on propensity score weighted Cox model has focused on the binary treatments. Generalized propensity score (GPS) \cite{imbens2000role} is the key device for causal inference with multiple treatments. Sugihara \cite{sugihara2010gps} investigated inverse probability weighted marginal Cox model, but did not incorporate the uncertainty in estimating the GPS-based weights in variance estimation. Cui et al. \cite{cui2018application} applied the marginal Cox model to a GPS-matched sample in a multi-group cancer study.} There has been a growing literature on propensity score weighting for multiple treatments with non-censored outcomes \cite{Feng2012,McCaffrey2013,li2019propensity} and related software packages, such as \texttt{PSweight}\cite{PSweight} and \texttt{twang},\cite{twang} but extension to time-to-event outcomes has been limited. 

A special case of multiple treatments is the factorial treatment design, in which two or more treatments and their combinations are studied simultaneously. For example, Simms et al. \cite{simms2024effect} studied the effect of the combination of two popular weight loss medications GLP-1 and sodium-glucose cotransporter-2 inhibitors compared to each medication alone, which constitutes a two-way or $2\times2$ factorial design. A $K-$ way factorial design has $2^K$ nominal treatments in total. The factorial design originates from randomized experiments, \cite{kahan2023reporting} but has been increasingly common in population based observational studies. Its analysis is often tied to ANOVA, which is not directly applicable to survival outcomes. Slud \cite{slud1994analysis} studied the asymptotic properties of the Cox model in randomized controlled trials with a two-way factorial design. 
VanderWeele proposed a covariate-adjusted interacted Cox model to study additive causal interaction effects between two treatments; \cite{vanderweele2011causal} {the coefficients of multivariable Cox models, however,} generally do not have a {population-level} causal interpretation. \cite{greenland1996absence, martinussen2013collapsibility}
Yu and Ding \cite{yu2025balancing} developed a class of balancing weights for causal inference in observational factorial studies, but their method assumes a linear outcome model and thus is not applicable to survival outcomes with censoring. 

In this paper, we combine the general propensity score weighting method for multiple treatments \cite{Feng2012,li2019propensity} and a marginal Cox model with indicators for each treatment to estimate the marginal hazard ratios between multiple treatments and a common reference treatment. {Our contributions are three-fold. First, we extend the inverse probability weights to the general class of balancing weights, including overlap weights and ATT weights (Section \ref{sec:Cox-multi}).} Second, we elaborate on the special case of two-way factorial treatments and point out its connection and distinction to the interacted Cox model (Section \ref{sec:Cox-factorial}). Third, we prove the consistency of the maximum weighted partial likelihood estimator and derive a robust sandwich variance estimator, {incorporating the uncertainty in estimating the generalized propensity scores} (Section \ref{sec:var-robust}). We conduct extensive simulations to evaluate the finite sample performance of the proposed method and compare two propensity weighting schemes: inverse probability weighting \cite{rosenbaum1987model} and overlap weighting.\cite{li2018balancing} We apply the proposed method to the motivating application aforementioned (Section \ref{sec:application}). To facilitate applied researchers, we develop an associated R package \texttt{PSsurvival}\cite{PSsurvival} for propensity score weighting analysis with time-to-event outcomes and provide an illustration example with data and code in the Supplementary Material.

\section{Propensity score weighted hazard ratio estimator for multiple treatments} \label{sec:model}
\subsection{Weighted marginal Cox model and estimator} \label{sec:Cox-multi}

Consider an observational study with $n$ units indexed by $i ~(= 1, \dots, n)$. We adopt the potential outcomes framework for causal inference. Each unit $i$ receives one of the $J+1$ ($J \geq 1$) treatments, denoted by $Z_i\in \mathbb{Z}=\{0,\ldots,J\}$, with $Z=0$ be the control or reference group. {Assuming no competing risk,} let $T_i(z)$ and $C_i(z)$ denote the potential survival and censoring time under treatment status $z\in \mathbb{Z}$, respectively, {of which the realized time are the ones corresponding to the received treatment $Z_i$, denoted as $T_{i}=T_{i}(Z_i)$ and $C_{i}=C_{i}(Z_i)$. For each unit $i$, only one of $T_{i}$ and $C_{i}$ can be observed due to censoring. We denote $U_i=\min(T_i,C_i)$ the observed time and $\delta_i = \one\{T_{i} \leq C_{i} \}$ the event indicator; $\delta_i=1$ if event and $0$ if censored.}
{Let $X_i$ the pre-treatment covariates. The generalized propensity score is the conditional probability of being assigned to each treatment group $z$ given the covariates: $e_{z}(x) \! = \! \Pr(Z=z\mid X=x)$ for all $z\in \mathbb{Z}$. \cite{imbens2000role} Denote $e_{i,z}=e_z(X_i)$ for unit $i$, then $\sum_{z \in \mathbb{Z}} e_{i,z}=1$ for each $i$.}

To define causal estimands, we denote $S_z(t) \! = \! \Pr(T(z) \geq t)$ and $\lambda_z(t) \! = \! -(\log S_z(t))'$ the (counterfactual) survival and hazard functions of treatment status $z$. Without any assumption beyond the existence of $\lambda_z(t)$, we define $\tau_z(t)$, the log marginal hazard ratio (MHR) between treatment $z \in \{1,...,J\}$ and $z=0$,
\begin{equation}\label{eq:hazard-ratio}
    \tau_z(t) := \log\{\lambda_z(t)/\lambda_0(t)\}.
\end{equation}
The trajectory of $\tau_z(t)$ over the study period is a complete description of the causal effect of treatment $z$ compared to the control on the hazard-ratio scale among the observed study population. However, instead of estimating the full trajectory of $\tau_z(t)$, it is often of practical interest to estimate and report a single scalar summary of $\tau_z(t)$ in medical research; \cite{hernan2000marginal, rulli2018assessment, kuitunen2022review} a dominating approach in practice is via the proportional hazards (PH) assumption: 
$\tau_z(t)$ is a constant over time, i.e., $\tau_z(t) \equiv \tau_z$ for all $t$.\cite{cox1972regression} 
Following the literature, we maintain the proportional hazard assumption for our developments throughout.

Under proportional hazards, the log MHR $\tau_z$ is equivalent to the coefficient of the indicator of treatment $z$ in the following marginal structural Cox model, \cite{sugihara2010gps, hernan2000marginal}
\begin{equation}
    \lambda_z(t)=\lambda_0 (t) \exp(\sum_{j=1}^J \tau_j \one\{j=z\}) \label{eq:cox_multi},
\end{equation}
Therefore, $\tau_z$ can be estimated via fitting Model \eqref{eq:cox_multi}, which is a marginal structural model \cite{hernan2000marginal} in the sense that it is imposed on the potential outcome rather than the observed outcome and is marginal over baseline covariates. However, directly fitting Model \eqref{eq:cox_multi} to the observed data $(U_i, \delta_i, Z_i)$  would lead to a biased estimate of $\tau_z$ because the distribution of $T(z)$ differs from that of the observed $T$ in the $Z=z$ arm due to confounding. The Cox model cannot be fitted with covariates either because of the non-collapsibility of hazard ratios. \cite{greenland1996absence, Struthers1986misspecify} Instead, we use propensity score weighting to adjust for confounding. We maintain three standard causal assumptions: (A1) SUTVA or consistency: $T_i=\sum_{z \in \mathbb{Z}} \one\{Z_i=z\} T_i(z)$ and $C_i=\sum_{z \in \mathbb{Z}} \one\{Z_i=z\} C_i(z)$; (A2) weak unconfoundedness:\cite{imbens2000role} $T(z)\perp \one\{Z=z\} \mid X,~~\forall~z\in \mathbb{Z}$; (A3) positivity or overlap: $e_z(X)>0$ for all $X$ and $z$. For the censoring mechanism, we assume (A4) treatment-dependent censoring: {$C(z) \perp\{T(z),X\} \mid Z=z$} for all $z$.

The propensity score weighted estimator of $\tau_z$ with weight $w$, denoted by $\hat{\tau}_z^w$, {is the solution to the following weighted partial likelihood estimating equations}, \cite{cox1975partial,lin1989robust,binder1992fitting}
\begin{equation}\label{eq:estimating-equation}
    \sum_{i=1}^n w_i \delta_i \left\{\one\{Z_i=z\} - \frac{\sum_{l \in \mathcal{R}_i} w_l \exp(\sum_{j=1}^J \one\{Z_l=j\} \hat{\tau}^w_j)\one\{Z_l=z\}}{\sum_{l \in \mathcal{R}_i} w_l \exp(\sum_{j=1}^J \one\{Z_l=j\} \hat{\tau}^w_j)}\right\}=0 ,\quad z=1,\cdots, J
\end{equation}
where {for unit $i$ with treatment $Z_i$ and observed time $U_i$, $\mathcal{R}_i \! = \! \{l \! : \! U_i\leq U_l \mid l \! = \! 1, \dots, n \}$} is the risk set {at $U_i$} and {$w_i$ is the propensity score weight}. The mechanism of using propensity score weighting is to re-weight the study sample to create a target population in which distributions of the covariates between treatments are balanced and thus eliminate confounding.\cite{li2018balancing} It may be tempting to combine propensity weighting and a covariate-adjusted multivariable Cox model to achieve double-robustness. However, Gabriel et al. \cite{gabriel2024propensity} show that such an approach generally does not lead to a doubly-robust estimator of the MHR.

Choice of the weighting scheme determines the target population, {and therefore, the estimand}. The most common choice is the inverse probability weight (IPW):\cite{rosenbaum1987model,Feng2012,McCaffrey2013} $w_{i,z} = {1}/{e_{i,z}}$ for each treatment $z\in\mathbb{Z}$. The IPW targets at the overall population, that is, the population represented by the entire study sample. A closely related weighting scheme, the ATT weight, $w_{i,z} = {e_{i,z'}}/{e_{i,z}}$ for each $z\in\mathbb{Z}$, akin to the average treatment effect for the treated (ATT) in binary treatments, targets at the subpopulation who receives a specific treatment $z'$. Another popular weighting scheme is the overlap weight (OW),\cite{li2019propensity} $w_{i,z}=(\sum_j  1/e_j(X_i))^{-1}/e_z(X_i)$ for each $z\in\mathbb{Z}$, which targets at the overlap population or patients in clinical equipoise, that is, the subpopulation with the most overlap in covariates between all treatments. {The weight for unit $i$ is $w_i \!=\! w_{i,Z_i}$, corresponding to the observed treatment $Z_i$. We denote by $\tau_z^w$ the log MHR for treatment $z$ over the target population of weighting scheme $w$. Under the proportional hazards assumption, $\tau_z^w$ coincides with the coefficient of treatment $z$ indicator in Model \eqref{eq:cox_multi}, had the study sample been drawn from it.}

{Different weighting schemes differ in their levels of sensitivity to poor overlap between treatment groups.} The IPW estimator is sensitive to propensities close to $0$, leading to large variances and inferior finite sample performance. Extreme propensities are more common in multiple treatments, in which poor overlap and extreme propensity scores are more prevalent. As illustrated in our application, often a small number of units with extreme propensities can dramatically change the result. A simple remedy is trimming, which removes subjects with extreme propensities, but it sometimes can result in a substantial loss of subjects and, more importantly, an opaque change in the target population. {Multiple treatments further complicate the selection of the threshold for trimming.} {The ATT weight also places propensity scores at the denominator and thus suffers from the same issue of extreme propensities as IPW.} In contrast, the OW gradually down-weights the subjects with small propensities instead of over-weighting them as IPW or ATT, and thus avoids the problem of extreme propensities. IPW, ATT, and OW all belong to the class of balancing weights,\cite{li2018balancing} among which OW leads to the smallest asymptotic total variance of all pairwise comparisons. \cite{li2019propensity} In practice, researchers should choose a specific weighting scheme according to their study goal, as well as the feature of the data, in particular, the degree of overlap. With poor overlap between groups, different weights can lead to markedly different results.

{The log MHR $\tau_z^w$ we aim for is the scalar summary of the trajectory $\tau_z^w(t)$ under proportional hazards in the target population, which is by far the most common in medical studies. In practice, researchers can test proportional hazards on the weighted sample via, e.g., diagnosis of Schoenfeld residuals or goodness-of-fit tests. \cite{schoenfeld1982residuals, Lin1991goodnessfit, grambsch1994testprophazard} In cases of violation, the weighted marginal Cox model estimates a censoring-dependent average of $\tau_z^w(t)$, \cite{xu2000AverageHazard,Luo2024drMHR} whose causal interpretation is controversial. \cite{martinussen2022causality, fay2024causal, Beyersmann2025hazards} We therefore do not recommend fitting the weighted marginal Cox model when PH fails, and suggest that researchers consider alternative estimands and methods that do not assume proportional hazards. \cite{mao2018propensity, cheng2021addressing, Luo2024drMHR}.}

For visualization, we can draw the weighted Kaplan-Meier survival curves of each treatment on the same plot. Specifically,  the weighted Kaplan-Meier curve for treatment $z$ is:    
$\widehat{S}^w_z(t) = \prod_{u_l \le t} \left(1 - \frac{D_l}{R_l}\right)$, with $D_l = \sum_{i=1}^n w_{i,z} \, \mathbf{1}\!\left(U_i = u_l, \delta_i = 1, Z_i=z\right)$ and $ R_l = \sum_{i=1}^n w_{i,z} \, \mathbf{1}\!\left(U_i \ge u_l, Z_i=z\right)$. The weighted accumulated risk plot is simply $1-\widehat{S}^w_z(t)$. Note that {weighting changes the sizes of the steps of Kaplan-Meier curves and thus their appearance from the unweighted ones.}

\subsection{A special case: two-way factorial treatments} \label{sec:Cox-factorial}
In a two-way factorial design, two treatments (A and B) are studied simultaneously, with all combinations of their levels included. The design allows for studying the main effect of each treatment as well as their interactive effects. Denote the observed treatment of each unit $i$ by two binary indicators $\mathbf{Z}_i \! = \! (Z_{i1},Z_{i2})$, where $Z_{i1},Z_{i2}\in \{0,1\}$. Despite the notational difference, a two-way factorial design can be viewed as a special case of multiple treatments with four nominal treatments: $\mathbf{z} \! = \! (z_1,z_2)\in \{0,1\}^2$. For each unit, there are four propensity scores, $\be_i \! = \! (e_{i,00},e_{i,01},e_{i,10},e_{i,11})$, corresponding to each of the four combinations, and $\sum_{\bz\in \{0,1\}^2}e_{i,\bz}=1$ for each $i$. 

For causal inference, the marginal Cox model \eqref{eq:cox_multi} for the potential survival time under treatment $\bz$ has the explicit form: 
\begin{equation} 
   \lambda_{\mathbf{z}}(t)=\lambda_0 (t) \exp\left\{\tau_1 \one\{\bz=(1,0)\} +\tau_2 \one\{\bz=(0,1)\}+ \tau_3\one\{\bz=(1,1)\}\right\} \label{eq:Cox_factor},
\end{equation}   
where \( \lambda_0(t) \) is the counterfactual baseline hazard function for the reference group $\mathbf{z}=(0,0)$, which consists of units who receive neither treatment, and \( (\tau_1, \tau_2, \tau_3)\) are the log MHRs between treatment A alone $\mathbf{z}=(1,0)$, treatment B alone $\mathbf{z}=(0,1)$, both treatments (AB) $\mathbf{z}=(1,1)$ and no treatment, respectively. Similarly as in Section \ref{sec:Cox-multi}, we can estimate these causal estimands on a target population by solving equation \eqref{eq:estimating-equation} with propensity score weights corresponding to that population, e.g., the IPW $w_{i,\bz} = {1}/{e_{i,\bz}}$ for the overall population and the OW {$w_{i,\bz} = {(\sum_{\bz^\prime \in \{0,1\}^2}  1/e_{\bz^\prime}(X_i))^{-1}}/{e_{i,\bz}}$} for the overlap population.

Model \eqref{eq:Cox_factor} is a reparametrization of the interacted Cox model previously proposed for $2\times2$ factorial randomized trials:\cite{slud1994analysis,vanderweele2011causal, lin2016simultaneous} 
\begin{equation}\label{eq:Cox-interact}
    \lambda_{\mathbf{z}}(t)=\lambda_0 (t) \exp(\beta_1 z_1+\beta_2 z_2+ \beta_3 z_1 z_2).
\end{equation}
with $\beta_1=\tau_1$, $\beta_2=\tau_2$, $\beta_3=\tau_3-\tau_1-\tau_2$. Because both models are saturated over the four unique treatment groups, they yield identical fits under the weighted partial-likelihood estimating equations. The variances of $(\hat{\beta}_1, \hat{\beta}_2, \hat{\beta}_3)$ can be obtained from those of $(\hat{\tau}_1, \hat{\tau}_2, \hat{\tau}_3)$. We adopt the parametrization by \eqref{eq:Cox_factor} instead of \eqref{eq:Cox-interact} because each of the coefficient is directly a log MHR against the reference group; in particular, $\tau_3$ is often itself an estimand of interest. In contrast, $\beta_3$ in Equation \eqref{eq:Cox-interact} measures the multiplicative interaction between treatments A and B on the hazard ratio, and does not correspond to a simple contrast between two treatment groups.

In theory, the two-way factorial design directly extends to $K$-way with $K>2$ treatments. {Specifically, Model \eqref{eq:Cox_factor} deploys $2^K-1$ indicators, where each $\tau_\bz$ is the log MHR of the group $\bz \!=\! (z_1,...,z_K) \in \{0,1\}^K$ compared to the reference group $\bz \!=\! (0,...,0)$. In contrast, the parametrization of \eqref{eq:Cox-interact} extends with $K$ main effects and all higher-order interactions, whose interpretations become increasingly harder as $K$ grows. In practice, most factorial designs focus on $K \!=\! 2$ given the rapid growth in the number of unique treatment combinations.} In the following, we will revert to the single-index treatment notation as in Section \ref{sec:Cox-multi}, unless otherwise specified.

\subsection{Consistency and robust estimator of variance} \label{sec:var-robust}

In this section, we show the consistency of the {weighted partial likelihood} estimator $\hat{\tau}_z^w$ and use the M-estimation technique \cite{lunceford04} to derive a robust sandwich estimator of its variance that accounts for the uncertainty in estimating the propensity scores (weights). \cite{lin1989robust,binder1992fitting,shu2021variance}  
{In practice, the generalized propensity score $e_z(x)$, and therefore, the propensity score weight $w$, is unknown and must be estimated. We consider a multinomial logistic propensity score model, $e_z(x;\gamma)=\exp(\gamma_z^{\prime} x)/(1+\sum_{j=1}^{J}\exp(\gamma_j^{\prime} x))$ for all $z \in \{1, \ldots, J \}$, where each $\gamma_j$ is the coefficient vector for treatment level $j$, and denote $\gamma=(\gamma_1, \ldots, \gamma_J)$ and its maximum likelihood estimator (MLE) $\hat{\gamma}=(\hat{\gamma}_1, \ldots, \hat{\gamma}_J)$. Denote $\mathbf{D}_i=(\one\{Z_i=1\},\ldots,\one\{Z_i=J\})^\prime$ and $\mathbf e(X_i;\gamma)=(e_1(X_i;\gamma),\ldots,e_J(X_i;\gamma))^\prime$ the vectors of treatment indicators and generalized propensity scores for unit $i$. The propensity score weight $w_i$ is determined by $\mathbf e(X_i;\gamma)$ and is therefore $w_i(\gamma)$. The estimated GPS and weights are $\mathbf{e}(X_i;\hat{\gamma})$ and $w_i(\hat{\gamma})$, respectively. For clarity, we further denote $\tau^w=(\tau_1^w, \ldots, \tau_J^w)^\prime$ the vector of log MHRs under weight $w$, and its weighted partial likelihood estimator $\hat{\tau}^w=(\hat{\tau}_1^w, \ldots, \hat{\tau}_J^w)^\prime$.}

{We define the at-risk indicator $Y_i(t)=\one\{U_i\ge t\}$ for unit $i$ at time $t$, which equals $1$ if at risk and $0$ otherwise. The weighted risk-set processes are defined as}
\begin{align*}
S^{(0)}(t;\tau^w, \gamma)
&=\frac{1}{n}\sum_{i=1}^n w_i(\gamma) Y_i(t)\exp({\eta_i(\tau^w)}),\\
\mathbf{S}^{(1)}(t;\tau^w, \gamma)
&=\frac{1}{n}\sum_{i=1}^n w_i(\gamma) Y_i(t)\mathbf{D}_i \exp({\eta_i(\tau^w)}),\\
\mathbf{S}^{(2)}(t;\tau^w, \gamma)
&=\frac{1}{n}\sum_{i=1}^n w_i(\gamma) Y_i(t)\mathbf{D}_i\mathbf{D}_i^{\prime} \exp({\eta_i(\tau^w)}),
\end{align*}
{where} $\eta_i(\tau^w) = (\tau^{w})^\prime \mathbf{D}_i=\sum_{j=1}^J \tau_j^w \one\{Z_i=j\}$. The weighted average treatment {vector} in the risk set {at time $t$} is
\begin{equation*}
    \bar{\mathbf{D}}(t;\tau^w, \gamma)=\frac{\mathbf{S}^{(1)}(t;\tau^w, \gamma)}{S^{(0)}(t;\tau^w, \gamma)}.
\end{equation*}
{The estimators} $(\hat{\tau}^w,\hat{\gamma})$ are the solutions to the joint estimating equations:
\begin{equation}
\sum_{i=1}^n
\phi_i(\tau^w,\gamma)=\sum_{i=1}^n
\begin{pmatrix}
\psi_i(\tau^w,\gamma)\\
\pi_i(\gamma)
\end{pmatrix}
=
\mathbf 0,
\label{eq:stacked}
\end{equation}
where $\psi_i(\tau^w,\gamma)=w_i(\gamma)\delta_i\{\mathbf{D}_i-\bar{\mathbf{D}}(U_i;\tau^w,\gamma)\}$ and $\pi_i(\gamma)=\big(\mathbf D_i-\mathbf e(X_i;\gamma)\big)\otimes X_i$.

\begin{theorem} \label{thm:variance}
Under Assumptions A1-A4 and standard regularity conditions, the weighted {partial likelihood} estimator $\hat{\tau}^w$ with a {propensity score} weight $w$ is consistent to the target log marginal hazard ratio $\tau^w$. The asymptotic covariance matrix of {$(\hat{\tau}^w, \hat\gamma)$} 
is consistently estimated by the robust sandwich estimator
\begin{equation}\label{eq:var-mat}
    \widehat{\mathrm{Var}}
\begin{pmatrix}
\hat{\tau}^w\\
\hat\gamma
\end{pmatrix}
{\frac{1}{n}\mathbf{A}(\hat\tau^w,\hat\gamma)^{-1} \,
\mathbf{B}(\hat\tau^w,\hat\gamma) \,
(\mathbf{A}(\hat\tau^w,\hat\gamma)^{-1})^\prime,}
\end{equation}
where 
\begin{equation*}
{\mathbf{A}(\tau^w,\gamma)} = -\frac{1}{n}\sum_{i=1}^n \frac{\partial}{\partial((\tau^{w})^\prime,\gamma^{\prime})} \phi_i(\tau^w,\gamma), \qquad
{\mathbf{B}(\tau^w,\gamma)} = \frac{1}{n}\sum_{i=1}^n \Phi_i(\tau^w,\gamma)\Phi_i(\tau^w,\gamma)^{\prime},
\end{equation*}
and
\begin{equation*}
\Phi_i(\tau^w,\gamma)=
\begin{pmatrix}
\Psi_i(\tau^w,\gamma)\\
\pi_i(\gamma)
\end{pmatrix},
\end{equation*}
and
\begin{align*}
\Psi_i(\tau^w,\gamma) =
w_i(\gamma) \delta_i \left\{\mathbf{D}_i-\bar{\mathbf{D}}(U_i;\tau^w,\gamma) \right\} -
w_i(\gamma)
\sum_{j=1}^n \frac{w_j(\gamma)\delta_j \, \one(U_j\le U_i)\exp({\eta_i(\tau^w)})}
     {nS^{(0)}(U_j;\tau^w,\gamma)}
\left\{ \mathbf{D}_i-\bar{\mathbf{D}}(U_j;\tau^w,\gamma) \right\}.
\end{align*}
\end{theorem}

Theorem \ref{thm:variance} {shows that the covariance matrix of $\hat\tau^w$ is consistently estimated by} the upper-left $J\times J$ block of the variance estimator matrix in \eqref{eq:var-mat}, {whose $J$ diagonal elements $\hat{\sigma}^2=(\hat{\sigma}_1^2, \ldots, \hat{\sigma}_J^2)$ are the estimated variances of $\hat{\tau}_1^w, \ldots, \hat{\tau}_J^w$.}{Correspondingly,} the 95\% confidence intervals of the {marginal} hazard ratios are $\exp({\hat{\tau}^w \pm 1.96 \hat{\sigma}})$, {where $\hat{\sigma} \!=\! (\hat{\sigma}^2)^{1/2} \!=\! (\hat{\sigma}_1, \ldots, \hat{\sigma}_J)$ are the estimated standard errors.} These confidence intervals are not symmetric because of the exponentiation. We can also {estimate} the standard errors of $\hat{\tau}^w$ by bootstrap. {For} binary treatments, the simulations of Austin \cite{austin2016performance} found that the bootstrap variance estimator often {outperformed} the robust sandwich variance estimator. However, when the event rate is low {or a treatment group is small}, some bootstrap samples might contain no events {or no study units from that group}, and thus cause computational issues.

\section{Simulation} \label{sec:simulation}
We conduct simulation studies to evaluate the finite-sample performance of the proposed estimators. We consider two {treatment scenarios, one with multiple treatments of three groups, and one with a two-way factorial design of two binary treatments}.

\subsection{Setting 1: Multiple Treatments}\label{sec:Simu-multi}
Our simulation design builds on that in Cheng et al. \cite{cheng2021addressing} and extends it to multiple treatments. Each unit has six baseline covariates. {We generated} $(X_1,X_2,X_3)$ from a multivariate normal distribution with mean zero, unit variance, and pairwise correlation $0.5$, and $(X_4,X_5,X_6)$ independently from a Bernoulli$(0.5)$ distribution, centered by subtracting $0.5$. {We generated the} treatment variable $Z\in\{0,1,2\}$ from a multinomial logistic model with propensity scores
\begin{equation*}
e_z(x)=\Pr(Z=z\mid X=x) =\frac{\exp\{g_z(x)\}}{\sum_{j=0}^2\exp\{g_j(x)\}},
\end{equation*}
{where $g_0(x) \!\equiv\! 0$, $g_1(x) \!=\! \alpha+ \kappa \, b^\prime x$, and $g_2(x) \!=\! \alpha - \kappa \, b^\prime x$}. The vector $b\!=\! (0.6, -0.4, 0.3, 0.2, -0.1, 0.15)^{\prime}$ is fixed and normalized to unit length, and the scalar $\kappa$ controls the degree of covariate overlap across treatment groups. We considered $\kappa \in\{1,2,3\}$, representing strong, moderate, and weak overlap. For each $\kappa$, {we chose} the intercept $\alpha$ {so that the three treatment groups had equal prevalence $1/3$}. {We visualize the distributions of propensity scores under the three degrees of overlap in the Supplementary Material}.

{We generated} the potential survival times {$\{T(z):z=0,1,2\}$} from a {Weibull model with conditional hazard}
\begin{equation} \label{eq:outcome-sim}
    \lambda_z(t\mid X\!=\!x)=h_0(t)\exp(\theta_z+\beta^\prime x),
\qquad \theta_0=0,
\end{equation}
where $h_0(t)$ is a Weibull baseline hazard with shape $1.2$ and scale $1.0$, the conditional log hazard ratios are $(\theta_1, \theta_2)\!=\! (0.35,-0.20)$, and the covariate coefficients are $\beta = (1.2,-0.9, 0.8, 0.6, -0.3, 0.4)^\prime$. We generated the censoring time independently from an exponential distribution $C \!\sim\! \mathrm{Exp}(\lambda_c)$, with $\lambda_c$ {calibrated to an overall censoring rate of} $25\%$ or $50\%$. The observed time is {$U \!=\! \min(T,C)$} and the event indicator is $\delta \!=\! \mathbf 1(T\le C)$. {Under strong overlap ($\kappa \!=\! 1$), 56\% and 45\% of the units experienced an event by $t \!=\! 1$ under $25\%$ and $50\%$ censoring, respectively. The event rates are similar across the degrees of overlap but differ across treatment groups, and we report them in the Supplementary Material.}

{Because hazard ratios are not collapsible over covariates, the conditional model \eqref{eq:outcome-sim} does not induce proportional marginal hazards, so $\tau_z(t)$ varies over time. We obtained the true value of the target estimand $\tau^w \!=\!(\tau_1^w,\tau_2^w)^\prime$ numerically, by solving the estimating equation \eqref{eq:estimating-equation} with weight $w$ on a large uncensored sample of the potential survival times from the same data generating mechanism. For the overall population, $(\tau_1^{ipw},\tau_2^{ipw}) \!=\! (0.17, - 0.10)$ under all three degrees of overlap. For the overlap population, $(\tau_1^{ow}, \tau_2^{ow}) \!=\! (0.20, -0.12)$, $(0.25, -0.15)$, and $(0.27, -0.17)$ under $\kappa \!=\! 1$, $2$, and $3$.}

For each configuration, we generated {$M\!=\! 1000$} replicates of sample size $n=1000$. In each replicate, we estimated the generalized propensity scores by a multinomial logistic model with the main effects of all covariates, {and estimated} the log MHRs {by the proposed estimator in \eqref{eq:estimating-equation}, using the inverse probability weights (IPW) for the overall population and the overlap weights (OW) for the overlap population. We estimated the variances by both the robust sandwich estimator and the bootstrap with 200 resamples.}
{For the overall population,} we also compared two alternatives,
(i) a naive Cox model {with only the} treatment indicators, and
(ii) a multivariable Cox model with the treatment indicators and all covariates. {Both estimate the coefficients of treatment indicators by unweighted partial likelihood, and take these coefficients, with their model-based variances, as estimates of $(\tau_1^{ipw},\tau_2^{ipw})$.}
{For each estimator, we report the relative bias $\{ M^{-1}\sum_{m} \exp({\hat{\tau}_{(m)}^w}) - \exp({\tau^w})\} / \exp({\tau^w})$, the average estimated standard errors, and the empirical coverage of the 95\% confidence intervals, where $\hat{\tau}_{(m)}^w$ is the estimate from the $m$-th replicate.}

\begin{table}[ht]
\centering
\caption{{Simulation results of the multiple treatment setting}.
Rel.Bias is the relative bias, and SE(ro) and SE(bs) are the {average estimated} standard errors from the robust sandwich estimator and bootstrap {over the $M$ replicates}, respectively.}
\label{tab:results_multi_1000}
\resizebox{\textwidth}{!}{%
\begin{tabular}{cccc*{2}{cccc}}
\toprule
Degree of  & Censor  & & \multicolumn{4}{c}{$\tau_1^w$} & \multicolumn{4}{c}{$\tau_2^w$}  \\
\cmidrule(lr){4-7} \cmidrule(lr){8-11} 
overlap   & Rate & Method & Rel. Bias & Coverage & SE (ro) & SE (bs) & Rel. Bias & Coverage & SE (ro) & SE (bs)  \\
\midrule
\multirow{4}{*}{Strong, $\kappa=1$} 
& \multirow{4}{*}{$25\%$} 
& IPW 
& 0.04 & 0.94 & 0.14 & 0.12
& 0.00 & 0.98 & 0.11 & 0.08 \\

& & OW 
& 0.02 & 0.97 & 0.11 & 0.10
& 0.00 & 0.98 & 0.09 & 0.08 \\

& & Naive 
& 0.79 & 0.00 & 0.09 & 0.09
& -0.42 & 0.00 & 0.10 & 0.10 \\

& & multivariable 
& 0.20 & 0.48 & 0.09 & 0.09
& -0.09 & 0.80 & 0.10 & 0.10 \\
\cmidrule(lr){1-11}
\multirow{4}{*}{Strong, $\kappa=1$} 
& \multirow{4}{*}{$50\%$} 
& IPW 
& 0.06 & 0.96 & 0.15 & 0.12
& -0.02 & 0.97 & 0.15 & 0.12 \\

& & OW 
& 0.06 & 0.97 & 0.13 & 0.11
& -0.01 & 0.98 & 0.13 & 0.11 \\

& & Naive 
& 0.89 & 0.00 & 0.10 & 0.10
& -0.46 & 0.00 & 0.13 & 0.13 \\

& & multivariable 
& 0.20 & 0.61 & 0.11 & 0.11
& -0.09 & 0.87 & 0.13 & 0.13 \\
\cmidrule(lr){1-11}
\multirow{4}{*}{Moderate, $\kappa=2$} 
& \multirow{4}{*}{$25\%$} 
& IPW 
& 0.16 & 0.77 & 0.22 & 0.21
& -0.02 & 0.90 & 0.17 & 0.15 \\

& & OW 
& 0.03 & 0.95 & 0.14 & 0.13
& 0.00 & 0.96 & 0.11 & 0.10 \\

& & Naive 
& 1.40 & 0.00 & 0.09 & 0.09
& -0.56 & 0.00 & 0.10 & 0.10 \\

& & multivariable 
& 0.20 & 0.52 & 0.10 & 0.10
& -0.09 & 0.82 & 0.11 & 0.11 \\

\cmidrule(lr){1-11}
\multirow{4}{*}{Moderate, $\kappa=2$} 
& \multirow{4}{*}{$50\%$} 
& IPW 
& 0.15 & 0.85 & 0.23 & 0.21
& -0.05 & 0.84 & 0.22 & 0.21 \\

& & OW 
& 0.05 & 0.95 & 0.16 & 0.15
& -0.01 & 0.96 & 0.15 & 0.15 \\

& & Naive 
& 1.53 & 0.00 & 0.10 & 0.10
& -0.60 & 0.00 & 0.14 & 0.14 \\

& & multivariable 
& 0.21 & 0.64 & 0.11 & 0.11
& -0.09 & 0.90 & 0.14 & 0.14 \\

\cmidrule(lr){1-11}
\multirow{4}{*}{Weak, $\kappa=3$} 
& \multirow{4}{*}{$25\%$} 
 & IPW 
& 0.35 & 0.62 & 0.26 & 0.25
& -0.12 & 0.76 & 0.20 & 0.19 \\

& & OW 
& 0.02 & 0.94 & 0.16 & 0.16
& 0.00 & 0.96 & 0.13 & 0.13 \\

& & Naive 
& 1.80 & 0.00 & 0.09 & 0.09
& -0.63 & 0.00 & 0.10 & 0.10 \\

& & multivariable 
& 0.20 & 0.57 & 0.10 & 0.10
& -0.09 & 0.85 & 0.11 & 0.11 \\

\cmidrule(lr){1-11}
\multirow{4}{*}{Weak, $\kappa=3$} 
& \multirow{4}{*}{$50\%$} 
& IPW 
& 0.33 & 0.70 & 0.28 & 0.26
& -0.18 & 0.70 & 0.27 & 0.27 \\

& & OW 
& 0.05 & 0.94 & 0.19 & 0.19
& 0.00 & 0.95 & 0.19 & 0.19 \\

& & Naive 
& 1.96 & 0.00 & 0.10 & 0.10
& -0.66 & 0.00 & 0.14 & 0.14 \\

& & multivariable 
& 0.21 & 0.69 & 0.12 & 0.12
& -0.09 & 0.90 & 0.15 & 0.15 \\
\bottomrule
\end{tabular}%
}
\end{table}

Table \ref{tab:results_multi_1000} summarizes the results. The naive Cox estimator {is severely biased with near-zero coverage in every scenario, which is expected because it ignores confounding and targets neither the marginal nor the conditional effect.} The multivariable Cox model improves over the naive estimator {but stays biased with poor coverage throughout}. Although the model is correctly specified for the conditional hazard, its coefficient of the treatment {indicator} generally does not recover the marginal {estimand} due to the non-collapsibility of {hazard ratios}.

{In contrast, the IPW estimator has small bias and coverage close to nominal under strong overlap. As the overlap deteriorates, its bias and variance increase sharply and its coverage drops well below nominal, a well-known limitation of IPW driven by extreme propensity weights. The OW estimator instead stays nearly unbiased for its target estimand across all degrees of overlap and censoring rates, with smaller variance and coverage near the nominal level. The robust variance estimator generally agrees with the bootstrap, which supports the validity of our closed-form estimator. It turns slightly conservative for IPW only as the overlap weakens, again reflecting the extreme weights.}

\subsection{Setting 2: Two-way factorial treatments}\label{sec:sim-factorial}

We next generated a two-way factorial design with two binary treatments $(Z_1,Z_2)$ and four treatment combinations $\bz\in\{(0,0),(1,0),(0,1),(1,1)\}$, {we followed the design of Section \ref{sec:Simu-multi} unless otherwise specified}. {The baseline covariates are as above, and we generated the treatment assignment from the multinomial logistic model}
\begin{equation} \label{eq:sim-factorial}
  e_{\bz}(x) = \Pr\{(Z_1,Z_2)\!=\!\bz \mid X\!=\!x \}
=\frac{\exp\{g_{\bz}(x)\}}{\sum_{\bz^\prime}\exp\{g_{\bz^\prime}(x)\}},  
\end{equation}
{where $g_{00}(x)\!\equiv\! 0$, $g_{10}(x) \!=\! \alpha_1+ \kappa\, b^\prime x$, $g_{01}(x)\!=\! \alpha_2- \kappa \, b^\prime x$, and $g_{11}(x)\!=\! \alpha_3+ \kappa \, c^\prime x$}.
{We kept} $b$ and set $c \!=\! (0.4,0.2,-0.3,0.1,0.1,-0.2)^\prime$, also normalized to unit length. {As before, $\kappa\in\{1,2,3\}$ controls the degree of overlap, and we calibrated $(\alpha_1,\alpha_2,\alpha_3)$ so that the four treatment combinations had equal prevalence $1/4$}.
{We display the distributions of the generalized propensity scores in the Supplementary Material}.

{We generated} the potential event times $T(\bz)$ {from a Weibull model with conditional hazard}
\begin{equation*}
\lambda_{\bz}(t\mid X\!=\!x) = h_0(t)\exp(\theta_{\bz}+\beta^\prime x), \qquad \theta_{00}=0,
\end{equation*}
{where $h_0(t)$ and $\beta$ are unchanged} and the conditional log hazard ratios are $(\theta_{10}, \theta_{01}, \theta_{11})\!=\! (0.35,-0.20,0.15)$. {We generated the censoring time $C \!\sim\! \mathrm{Exp}(\lambda_c)$ to target} $25\%$ or $50\%$ censoring. As before, we calculated the true estimands numerically. {For the overall population, the log MHRs of groups} $10, 01, 11$ against the reference group $00$ are $(\tau_1^{ipw},\tau_2^{ipw},\tau_3^{ipw}) \!=\! (0.17, -0.10, 0.07)$. {For the overlap population, $(\tau_1^{ow},\tau_2^{ow},\tau_3^{ow}) \!=\! (0.20,-0.12,0.08)$, $(0.24, -0.14, 0.10)$, and $(0.26, -0.16, 0.11)$ under $\kappa\!=\! 1$, $2$, and $3$.} We used the same number of replicates, sample sizes, and {estimators within each replicate} as in Section \ref{sec:Simu-multi}.

\begin{table}[ht]
\centering
\caption{{Simulation results of the two-way factorial design setting}. {Rel.Bias is the relative bias, and} SE(ro) and SE(bs) are the {average estimated} standard errors from the robust sandwich estimator and bootstrap {over the $M$ replicates}, respectively.}
\label{tab:results_factor_1000}
\resizebox{\textwidth}{!}{%
\begin{tabular}{cccc*{3}{cccc}}
\toprule
Degree of  & Censor  & & \multicolumn{4}{c}{$\tau_1^w$} & \multicolumn{4}{c}{$\tau_2^w$} & \multicolumn{4}{c}{$\tau_3^w$} \\
\cmidrule(lr){4-7} \cmidrule(lr){8-11} \cmidrule(lr){12-15}
overlap   & Rate & Method 
& Rel.Bias & Coverage & SE(ro) & SE(bs) 
& Rel.Bias & Coverage & SE(ro) & SE(bs)
& Rel.Bias & Coverage & SE(ro) & SE(bs) \\
\midrule
\multirow{4}{*}{Strong, $\kappa=1$}
& \multirow{4}{*}{$25\%$}

& IPW 
& 0.04 & 0.95 & 0.15 & 0.13
& 0.00 & 0.99 & 0.13 & 0.10
& 0.02 & 0.98 & 0.13 & 0.10 \\

& & OW 
& 0.03 & 0.96 & 0.13 & 0.11
& 0.00 & 0.98 & 0.11 & 0.09
& 0.01 & 0.97 & 0.11 & 0.09 \\

& & Naive 
& 0.82 & 0.00 & 0.10 & 0.10
& -0.42 & 0.00 & 0.11 & 0.11
& 0.15 & 0.74 & 0.10 & 0.10 \\

& & multivariable 
& 0.21 & 0.57 & 0.10 & 0.10
& -0.08 & 0.87 & 0.11 & 0.11
& 0.09 & 0.88 & 0.11 & 0.11 \\

\cmidrule(lr){1-15}
\multirow{4}{*}{Strong, $\kappa=1$} 
& \multirow{4}{*}{$50\%$} 

&  IPW 
& 0.06 & 0.95 & 0.17 & 0.14
& -0.02 & 0.98 & 0.18 & 0.16
& 0.03 & 0.98 & 0.15 & 0.12 \\

& & OW 
& 0.05 & 0.95 & 0.15 & 0.13
& -0.02 & 0.97 & 0.15 & 0.14
& 0.03 & 0.96 & 0.14 & 0.13 \\

& & Naive 
& 0.93 & 0.00 & 0.12 & 0.12
& -0.48 & 0.01 & 0.15 & 0.15
& 0.18 & 0.76 & 0.13 & 0.13 \\

& & multivariable 
& 0.21 & 0.69 & 0.12 & 0.12
& -0.08 & 0.90 & 0.15 & 0.15
& 0.09 & 0.91 & 0.13 & 0.13 \\

\cmidrule(lr){1-15}
\multirow{4}{*}{Moderate, $\kappa=2$} 
& \multirow{4}{*}{$25\%$} 

& IPW 
& 0.16 & 0.77 & 0.23 & 0.22
& -0.03 & 0.89 & 0.19 & 0.18
& 0.05 & 0.95 & 0.19 & 0.17 \\

& & OW 
& 0.05 & 0.94 & 0.17 & 0.17
& 0.00 & 0.97 & 0.14 & 0.13
& 0.02 & 0.97 & 0.13 & 0.13 \\

& & Naive 
& 1.53 & 0.00 & 0.10 & 0.10
& -0.57 & 0.00 & 0.12 & 0.12
& 0.19 & 0.62 & 0.10 & 0.10 \\

& & multivariable 
& 0.21 & 0.61 & 0.11 & 0.11
& -0.08 & 0.88 & 0.12 & 0.12
& 0.09 & 0.90 & 0.11 & 0.11 \\

\cmidrule(lr){1-15}
\multirow{4}{*}{Moderate, $\kappa=2$} 
& \multirow{4}{*}{$50\%$} 

& IPW 
& 0.16 & 0.85 & 0.25 & 0.23
& -0.07 & 0.83 & 0.27 & 0.27
& 0.06 & 0.96 & 0.21 & 0.19 \\

& & OW 
& 0.07 & 0.93 & 0.19 & 0.19
& -0.02 & 0.96 & 0.20 & 0.20
& 0.03 & 0.96 & 0.17 & 0.17 \\

& & Naive 
& 1.70 & 0.00 & 0.12 & 0.12
& -0.62 & 0.00 & 0.17 & 0.17
& 0.24 & 0.64 & 0.13 & 0.13 \\

& & multivariable 
& 0.21 & 0.70 & 0.13 & 0.13
& -0.09 & 0.91 & 0.17 & 0.17
& 0.09 & 0.91 & 0.14 & 0.14 \\

\cmidrule(lr){1-15}
\multirow{4}{*}{Weak, $\kappa=3$} 
& \multirow{4}{*}{$25\%$} 

& IPW 
& 0.36 & 0.65 & 0.29 & 0.29
& -0.14 & 0.74 & 0.23 & 0.23
& 0.10 & 0.89 & 0.25 & 0.24 \\

& & OW 
& 0.08 & 0.90 & 0.22 & 0.23
& 0.01 & 0.95 & 0.18 & 0.18
& 0.03 & 0.94 & 0.17 & 0.17 \\

& & Naive 
& 2.06 & 0.00 & 0.10 & 0.10
& -0.64 & 0.00 & 0.12 & 0.12
& 0.18 & 0.64 & 0.10 & 0.10 \\

& & multivariable 
& 0.21 & 0.64 & 0.11 & 0.11
& -0.08 & 0.90 & 0.13 & 0.13
& 0.09 & 0.90 & 0.12 & 0.12 \\

\cmidrule(lr){1-15}
\multirow{4}{*}{Weak, $\kappa=3$} 
& \multirow{4}{*}{$50\%$} 

& IPW 
& 0.33 & 0.73 & 0.31 & 0.31
& -0.20 & 0.70 & 0.32 & 0.34
& 0.11 & 0.93 & 0.28 & 0.27 \\

& & OW 
& 0.10 & 0.91 & 0.25 & 0.26
& 0.00 & 0.94 & 0.27 & 0.28
& 0.04 & 0.95 & 0.22 & 0.22 \\

& & Naive 
& 2.28 & 0.00 & 0.12 & 0.12
& -0.68 & 0.00 & 0.17 & 0.17
& 0.25 & 0.60 & 0.13 & 0.13 \\

& & multivariable 
& 0.21 & 0.72 & 0.14 & 0.14
& -0.08 & 0.91 & 0.18 & 0.18
& 0.09 & 0.91 & 0.15 & 0.15 \\

\bottomrule
\end{tabular}%
}
\end{table}

Table \ref{tab:results_factor_1000} reports the results for the factorial design, {which agree with those of the multiple treatment setting}.
{The naive and multivariable Cox models are again substantially biased with poor coverage. The OW estimator stays nearly unbiased across all degrees of overlap, while the IPW estimator deteriorates as the overlap weakens. This holds for all three log MHRs, including $\tau_3^w$ of the combination group, which our parametrization reports directly.}

\section{Application: Real-world effectiveness of GLP-1 on heart failure} \label{sec:application}

We apply the proposed method to a dataset constructed from the Truveta Database, a {collection of} de-identified electronic health records from over 30 health systems in the United States, to evaluate the real-world effectiveness of glucagon-like peptide-1 receptor agonists (GLP-1) on heart failure endpoints in new users with heart failure with preserved ejection fraction (HFpEF). Our study cohort {includes 94,313 patients who initiated an anti-obesity medication between 2009 and 2024.} The primary outcome is the time to the first occurrence of a composite heart failure (HF) endpoint, defined as either HF hospitalization, urgent HF encounter, or cardiovascular death within 1 year. There are three treatments: GLP-1 (29,375 patients), naltrexone-bupropion (NB) (45,099 patients), and phentermine-topiramate (PT) (19,839 patients).

Baseline covariates include demographic characteristics (age, sex, race and ethnicity), laboratory measurements (BMI, blood pressure, left ventricular ejection fraction (LVEF), HbA1c, cholesterol), and {various} comorbidities, including hypertension, renal disease, pulmonary disease, and mental health conditions. Table~\ref{tab:HFpEF-baseline} {summarizes some} important baseline characteristics by treatment, {with} the standardized mean difference (SMD) of NB versus GLP-1 and PT versus GLP-1 before and after overlap weighting. {The three treatment groups are substantially imbalanced in age, sex, race, BMI, laboratory measures, and comorbidities, all with SMD greater than 0.1}. In particular, the GLP-1 group has a {much} higher {median} BMI (37.3) and a higher percentage of obesity (76\%) than the NB (31.9 and 38\%, respectively) and PT (34.5 and 53\%, respectively) groups. Overlap weighting improves balance in all covariates, reducing the SMDs to {around 0.05 or below}. More details of the clinical background and study design can be found in Kalapura et al. \cite{kalapura2026real}  

\begin{table}[ht]
\centering
\caption{Summary of baseline characteristics of patients by treatment and pairwise standardized mean difference in the HFpEF study.}
\label{tab:HFpEF-baseline}
\resizebox{\textwidth}{!}{%
\begin{tabular}{lrrrrrrr}
\toprule
\multirow{2}{*}{\textbf{Characteristic}}
& \multirow{3}{*}{\shortstack{GLP-1 RA\\(n=29{,}375)}}
& \multirow{3}{*}{\shortstack{NB\\(n=45{,}099)}}
& \multirow{3}{*}{\shortstack{PT\\(n=19{,}839)}}
& \multicolumn{2}{c}{SMD (GLP vs NB)}
& \multicolumn{2}{c}{SMD (GLP vs PT)} \\
\cmidrule(lr){5-6}\cmidrule(lr){7-8}
& & & & Unweighted & OW & Unweighted & OW \\
\midrule

\multicolumn{8}{l}{\textbf{Demographics}} \\
Age, Median (Q1, Q3)
& 58 (48, 67) & 60 (47, 70) & 52 (41, 63) & -0.08 & -0.01 & 0.36 & -0.03 \\
Female, No. (\%)
& 18{,}759 (64) & 27{,}095 (60) & 15{,}206 (76) & 0.08 & 0.02 & -0.26 & 0.05 \\
White, No. (\%)
& 22{,}509 (77) & 37{,}361 (83) & 14{,}539 (73) & -0.16 & -0.02 & 0.08 & -0.01 \\
Black, No. (\%)
& 4{,}221 (14) & 4{,}075 (9) & 3{,}497 (18) & 0.17 & 0.01 & -0.09 & -0.01 \\

\midrule
\multicolumn{8}{l}{\textbf{Laboratory measures}} \\
LVEF, Median (Q1, Q3)
& 64 (60, 68) & 63 (58, 68) & 63 (59, 68) & 0.13 & 0.00 & 0.06 & 0.01 \\
HbA1c, Median (Q1, Q3)
& 5.6 (5.3, 5.9) & 5.5 (5.3, 5.8) & 5.5 (5.2, 5.8) & 0.20 & -0.01 & 0.27 & 0.00 \\
HDL Cholesterol, Median (Q1, Q3)
& 50 (41, 61) & 50 (40, 62) & 49 (41, 60) & 0.00 & 0.02 & 0.05 & 0.02 \\
LDL Cholesterol, Median (Q1, Q3)
& 99 (75, 124) & 99 (75, 124) & 102 (80, 126) & -0.01 & 0.02 & -0.09 & 0.02 \\
White Blood Cell Count, Median (Q1, Q3)
& 7.0 (5.8, 8.6) & 7.4 (5.9, 9.2) & 7.3 (5.9, 9.0) & -0.16 & 0.01 & -0.12 & 0.02 \\
Hemoglobin, Median (Q1, Q3)
& 13.7 (12.8, 14.7) & 13.4 (12.2, 14.6) & 13.4 (12.3, 14.3) & 0.25 & -0.01 & 0.28 & -0.01 \\
Systolic Blood Pressure, Median (Q1, Q3)
& 156 (140, 173) & 158 (142, 177) & 155 (140, 174) & -0.11 & -0.02 & 0.00 & -0.01 \\
Diastolic Blood Pressure, Median (Q1, Q3)
& 94 (86, 103) & 95 (86, 106) & 95 (86, 105) & -0.10 & -0.02 & -0.07 & -0.01 \\
BMI, Median (Q1, Q3)
& 37.3 (32.8, 43.2) & 31.9 (28.6, 37.0) & 34.5 (30.1, 40.5) & 0.65 & 0.00 & 0.30 & 0.06 \\

\midrule
\multicolumn{8}{l}{\textbf{Comorbidities}} \\
Cardiac Arrhythmias, No. (\%)
& 11{,}561 (39) & 17{,}597 (39) & 6{,}270 (32) & 0.01 & -0.01 & 0.16 & -0.01 \\
Valvular Disease, No. (\%)
& 13{,}368 (46) & 19{,}332 (43) & 6{,}217 (31) & 0.05 & 0.02 & 0.29 & 0.01 \\
Peripheral Vascular Disorders, No. (\%)
& 15{,}407 (52) & 22{,}817 (51) & 10{,}773 (54) & 0.04 & 0.05 & -0.04 & 0.00 \\
Hypertension, Uncomplicated, No. (\%)
& 19{,}547 (67) & 26{,}247 (58) & 10{,}207 (51) & 0.17 & -0.03 & 0.31 & -0.02 \\
Other Neurological Disorders, No. (\%)
& 18{,}770 (64) & 29{,}618 (66) & 14{,}019 (71) & -0.04 & -0.03 & -0.15 & 0.00 \\
Chronic Pulmonary Disease, No. (\%)
& 7{,}277 (25) & 13{,}673 (30) & 5{,}742 (29) & -0.12 & 0.00 & -0.09 & 0.01 \\
Obesity, No. (\%)
& 22{,}243 (76) & 16{,}927 (38) & 10{,}431 (53) & 0.84 & -0.01 & 0.50 & -0.01 \\
Fluid and Electrolyte Disorders, No. (\%)
& 21{,}425 (73) & 29{,}838 (66) & 11{,}575 (58) & 0.15 & -0.02 & 0.31 & -0.01 \\
Drug Abuse, No. (\%)
& 9{,}434 (32) & 21{,}127 (47) & 6{,}664 (33) & -0.31 & -0.02 & -0.03 & -0.02 \\
Depression, No. (\%)
& 5{,}351 (18) & 18{,}112 (40) & 5{,}099 (26) & -0.50 & 0.00 & -0.18 & 0.02 \\
\bottomrule
\end{tabular}
}
\end{table}

We use a multinomial logistic model with a main effect of each baseline characteristic to estimate the generalized propensity scores for the three treatment groups, the distributions of which are shown in Figure~\ref{fig:HFpEF_ps_overlap}. The clear difference {between the propensity score distributions reflects} the covariate imbalance between the treatment groups. There are {several} subjects with extreme propensities in each {group}. In the GLP-1 group, for example, {one patient has a propensity to GLP-1 of $10^{-7}$ and two others have propensities to GLP-1 below 0.001, all translating into enormous inverse-propensity weights}. {If included, these few subjects} would distort the {IPW} analysis and lead to an estimated MHR of over 50. Trimming subjects with propensities less than 0.05 {excludes 9372 patients (9.9\% of the cohort)}, consisting of 702 GLP-1, 7110 NB, and 1560 PT patients. We fitted the weighted Cox model with OW {and} IPW, {the latter} with and without trimming.
\begin{figure}
    \centering
    \includegraphics[width=0.8\linewidth]{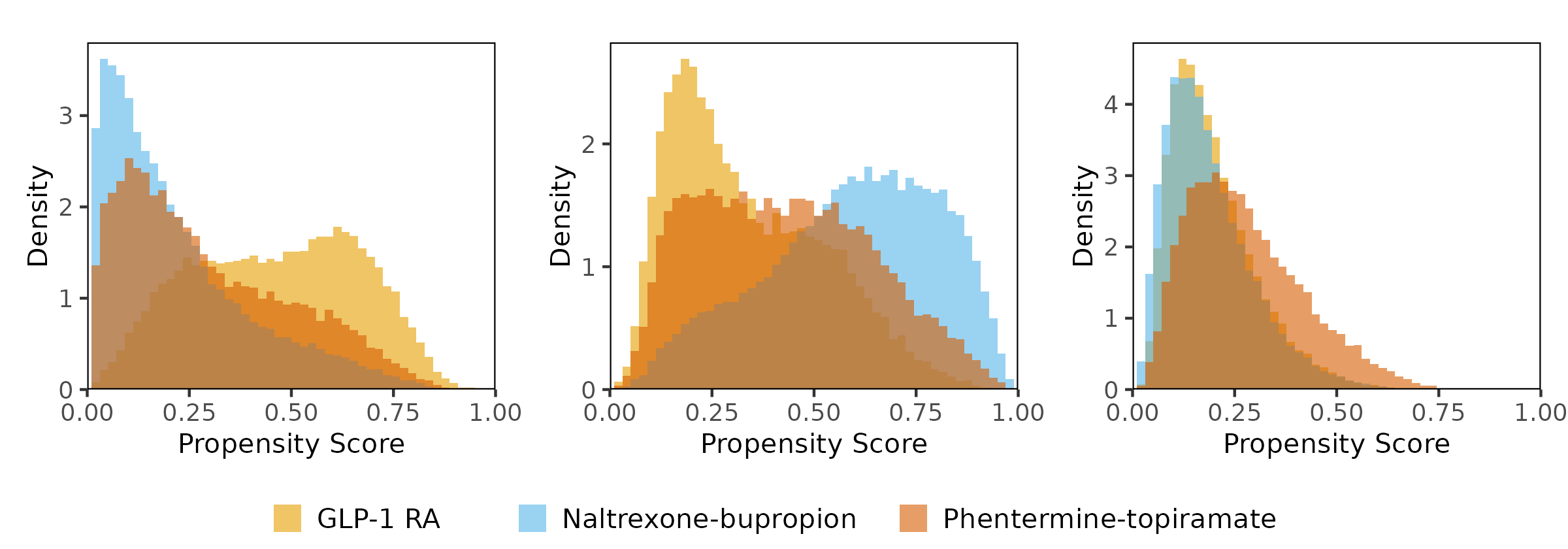}
    \caption{Distributions of the estimated generalized propensity scores {per treatment group} in the HFpEF application. {The panels (from left to right) are: GLP-1, NB, and PT.}}
    \label{fig:HFpEF_ps_overlap}
\end{figure}

Figure \ref{fig:HFpEF-KM} shows the unadjusted and overlap weighted Kaplan-Meier curves by treatment group. In the unadjusted analysis, the survival curves of GLP-1 and PT {groups} are similar over time, and both {consistently stay above the curve of the NB group}. After {overlap weighting, the three curves separate more clearly and stably, with the survival probability highest for the GLP-1 group}, followed by the PT group, and the lowest for NB throughout the follow-up.
\begin{figure}[ht]
    \centering
    \begin{minipage}[t]{0.49\linewidth}
        \centering
        \includegraphics[width=0.9\linewidth, height=2.5in]{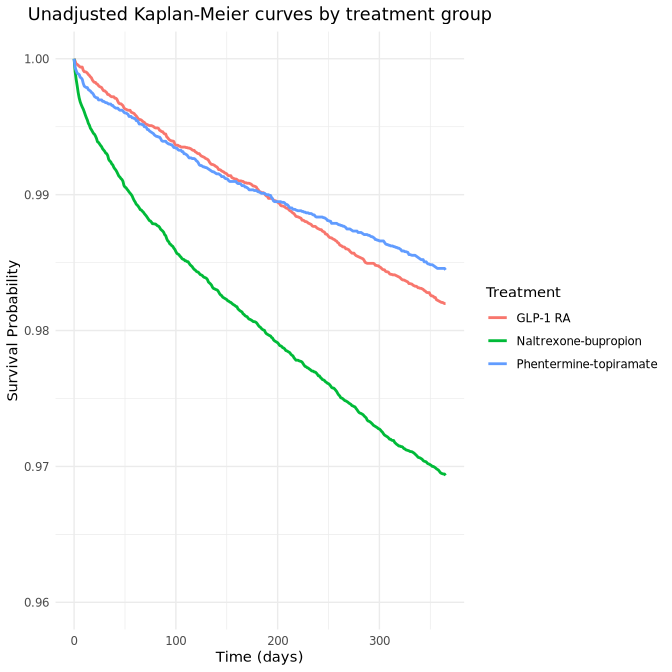}
    \end{minipage}
    \hspace{-0.05\linewidth}
    \begin{minipage}[t]{0.49\linewidth}
        \centering
        \includegraphics[width=0.9\linewidth, height=2.5in]{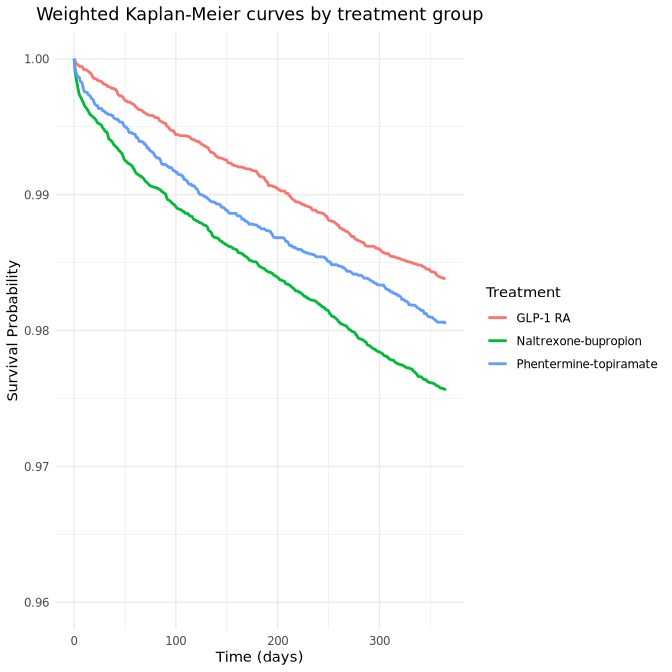}
    \end{minipage}
    \caption{Kaplan-Meier Curves by treatment of the HFpEF study: unadjusted (left) and overlap weighted (right).}
    \label{fig:HFpEF-KM}
\end{figure}

We {estimate} the MHR for the composite HF endpoint with GLP-1 as the reference group. Both the PT and NB are associated with higher hazards relative to GLP-1, with {OW} estimated MHRs of 1.52 (95\% CI: 1.34–1.73) and 1.21 (95\% CI: 1.03–1.42), {respectively}. These results are consistent with the {ordering of the three OW-weighted} Kaplan–Meier curves. As expected, the few patients with extreme propensities lead to extremely large and unstable IPW estimates, 54.82 (95\% CI: 8.11--370.38) for NB and 44.05 (95\% CI: 6.50--298.51) for PT.
With trimming, the IPW estimates become more stable, with {MHRs} of 1.47 (95\% CI: 1.29–1.68) for NB and 1.26 (95\% CI: 1.06–1.50) for PT, {aligning closely} with the OW estimates but with wider confidence intervals due to the loss of patients. {Together}, these results suggest a statistically significant benefit of the GLP-1 on heart failure endpoints among HFpEF patients compared to other anti-obesity medications.

\section{Discussion}\label{sec:discussion}

This paper developed the propensity-score-weighted Cox model for causal inference with survival outcomes in observational studies with multiple {or factorial} treatments. We provide the estimation and inference procedure for the marginal hazard ratio (MHR) estimand. For the special case of factorial studies, {our model is a reparametrization of the conventional interacted Cox model, but with a clearer interpretation, as each coefficient is directly a log MHR against the reference group.} Between the propensity score methods, we adopted the weighting approach because its extension to multiple treatments is easy to implement. In contrast, in binary treatments, a popular approach is to first obtain a sample via propensity score matching and then run a marginal Cox model on the matched sample. However, matching {becomes challenging as the number of treatment groups increases}. Moreover, poor overlap is common with {multiple} treatments, in which case matching can exclude a large portion of units, resulting in loss of power. The same problem applies to IPW with trimming, but is bypassed by OW. We assumed {treatment-dependent} censoring. {It is straightforward to relax this to covariate-dependent censoring, $C(z) \!\perp\! T(z) \mid X, Z \!=\! z$, which, however, would reintroduce confounding and thus require an additional inverse probability of censoring weight (IPCW)}. We can multiply the {IPCW} with the propensity score (treatment) weight to create a combined weight, \cite{hernan2000marginal,cheng2021addressing} but the combined weight would be sensitive to extreme probabilities of {both} treatment {and} censoring.

{We defined the MHR as the ratio between the counterfactual hazard functions of two treatments, and invoked the proportional hazards (PH) assumption to reduce its trajectory to a constant scalar. We acknowledge the ongoing debate on the causal interpretation of the marginal Cox model, \cite{hernan2010hazards, martinussen2022causality,fay2024causal,Luo2024drMHR,Beyersmann2025hazards} and caution against its use when the PH assumption is violated.} {Regardless of the conceptual debate, the hazard ratio remains the most widely reported estimand in clinical research with survival outcomes, and the Cox PH model is the dominant tool to estimate it. Importantly, the PH assumption is testable prior to the analysis. When it is violated, researchers can estimate the trajectory of the hazard ratio over time or another scalar summary of it. \cite{Luo2024drMHR} Alternatively, researchers can report estimands in the additive scale, such as differences in counterfactual survival probabilities \cite{cheng2021addressing} or the restricted mean survival time. \cite{mao2018propensity}}

Weighted Cox models can also be used for covariate adjustment in randomized trials to increase efficiency by adjusting for chance imbalance, \cite{FDAcovadj2023} {where} the choice of weights does not change the target population.\cite{zeng2021propensity} This is in contrast to covariate adjustment in observational studies, which serves to eliminate confounding bias. Despite the different goals, the same analytical model and inference procedure apply to both randomized trials and observational studies.

To bridge the proposed method to practice, we develop an associated R package \texttt{PSsurvival}\cite{PSsurvival} to implement general propensity score weighting analysis with time-to-event outcomes, available on CRAN \url{https://CRAN.R-project.org/package=PSsurvival}. Besides the marginal hazard ratios, the package also supports estimation and inference of counterfactual survival functions, with IPW (with and without trimming), OW, and ATT weighting, for both binary and multiple (including factorial) treatments. It supplies visualization functions to plot weighted Kaplan-Meier and cumulative risk curves. We provide an illustration example with publicly available data from a two-way factorial randomized trial \cite{lin2016simultaneous} and the R code in the Supplementary Material, which serves as a template for {analyzing} a multiple or factorial observational study using the \texttt{PSsurvival} package.

\section*{Acknowledgments}
This research is partially supported by a Duke University grant to the Observational Research Building Interdisciplinary Therapeutic Advances (ORBIT) Hub and a grant (25GLP1450119) from the American Heart Association (AHA). The content of this article is solely the responsibility of the authors and does not necessarily represent the view of Duke University or AHA. The authors appreciate the clinical input and motivating questions from Jay Lusk, Brian Mac Grory, computing assistance from Jiwon Shin, and constructive comments from Fan Li of Yale University and Per Johansson. 

\bibliographystyle{unsrtnat}
\bibliography{factorial}

@article{shu2021variance,
  title={Variance estimation in inverse probability weighted Cox models},
  author={Shu, Di and Young, Jessica G and Toh, Sengwee and Wang, Rui},
  journal={Biometrics},
  volume={77},
  number={3},
  pages={1101--1117},
  year={2021},
  publisher={Oxford University Press}
}

@article{austin2014use,
  title={The use of propensity score methods with survival or time-to-event outcomes: reporting measures of effect similar to those used in randomized experiments},
  author={Austin, Peter C},
  journal={Statistics in medicine},
  volume={33},
  number={7},
  pages={1242--1258},
  year={2014},
  publisher={Wiley Online Library}
}

@article{austin2017performance,
  title={The performance of inverse probability of treatment weighting and full matching on the propensity score in the presence of model misspecification when estimating the effect of treatment on survival outcomes},
  author={Austin, Peter C and Stuart, Elizabeth A},
  journal={Statistical methods in medical research},
  volume={26},
  number={4},
  pages={1654--1670},
  year={2017},
  publisher={SAGE Publications Sage UK: London, England}
}

@article{cheng2021addressing,
  title={Addressing Extreme Propensity Scores in Estimating Counterfactual Survival Functions via the Overlap Weights},
  author={Cheng, Chao and Li, Fan and Thomas, Laine and Li, Fan},
  journal={American Journal of Epidemiology},
  year={2021}
}

@article{li2018balancing,
   author = {Li, F. and Morgan, K. L. and Zaslavsky, A. M.},
   title = {Balancing Covariates via Propensity Score Weighting},
   journal = {Journal of the American Statistical Association},
   volume = {113},
   number = {521},
   pages = {390-400},
   ISSN = {0162-1459},
   year = {2018},
   type = {Journal Article}
}

@article{li2019propensity,
  title={Propensity score weighting for causal inference with multiple treatments},
  author={Li, Fan and Li, Fan},
  year={2019},
journal = {Annals of Applied Statistics},
   volume = {13},
   number = {4},
   pages = {2389-2415},

}

@article{lunceford04,
   author = {Lunceford, Jared K and Davidian, Marie},
   title = {Stratification and weighting via the propensity score in estimation of causal treatment effects: a comparative study},
   journal = {Statistics in medicine},
   volume = {23},
   number = {19},
   pages = {2937-2960},
   ISSN = {0277-6715},
   year = {2004},
   type = {Journal Article}
}

@article{lin1989robust,
author = {D. Y. Lin and L. J. Wei},
title = {The Robust Inference for the Cox Proportional Hazards Model},
journal = {Journal of the American Statistical Association},
volume = {84},
number = {408},
pages = {1074--1078},
year = {1989},
publisher = {ASA Website},
}

@article{mao2018propensity,
  title={On the propensity score weighting analysis with survival outcome: Estimands, estimation, and inference},
  author={Mao, Huzhang and Li, Liang and Yang, Wei and Shen, Yu},
  journal={Statistics in medicine},
  volume={37},
  number={26},
  pages={3745--3763},
  year={2018},
  publisher={Wiley Online Library}
}

@article{zeng2021propensity,
  title={Propensity score weighting for covariate adjustment in randomized clinical trials},
  author={Zeng, Shuxi and Li, Fan and Wang, Rui and Li, Fan},
  journal={Statistics in medicine},
  volume={40},
  number={4},
  pages={842--858},
  year={2021},
  publisher={Wiley Online Library}
}

@article{austin2016performance,
  title={The performance of different propensity score methods for estimating absolute effects of treatments on survival outcomes: a simulation study},
  author={Austin, Peter C and Schuster, Tibor},
  journal={Statistical Methods in Medical Research},
  volume={25},
  number={5},
  pages={2214--2237},
  year={2016},
  publisher={SAGE Publications Sage UK: London, England}
}

@article{hernan2010hazards,
  title={The hazards of hazard ratios},
  author={Hern{\'a}n, Miguel A},
  journal={Epidemiology},
  volume={21},
  number={1},
  pages={13},
  year={2010},
  publisher={NIH Public Access}
}

@ARTICLE{rosenbaum1983central,
   author  = {Rosenbaum, P R and Rubin, D B},
   TITLE   = {The Central Role of the Propensity Score in Observational Studies for Causal Effects},
   YEAR    = {1983},
   JOURNAL = {Biometrika},
   VOLUME  = {70},
   NUMBER = {1},
   PAGES   = {41--55}
}

@article{vanderweele2011causal,
  title={Causal interactions in the proportional hazards model},
  author={VanderWeele, Tyler J},
  journal={Epidemiology},
  volume={22},
  number={5},
  pages={713--717},
  year={2011},
  publisher={LWW}
}

@article{imbens2000role,
  title={The role of the propensity score in estimating dose-response functions},
  author={Imbens, Guido W},
  journal={Biometrika},
  volume={87},
  number={3},
  pages={706--710},
  year={2000},
  publisher={Oxford University Press}
}

@article{Feng2012,
author = {Feng, Ping and Zhou, Xiao Hua and Zou, Qing Ming and Fan, Ming Yu and Li, Xiao Song},
journal = {Statistics in Medicine},
number = {7},
pages = {681--697},
pmid = {21351291},
title = {{Generalized propensity score for estimating the average treatment effect of multiple treatments}},
volume = {31},
year = {2012}
}

@article{McCaffrey2013,
author = {McCaffrey, Daniel F. and Griffin, Beth Ann and Almirall, Daniel and Slaughter, Mary Ellen and Ramchand, Rajeev and Burgette, Lane F.},
journal = {Statistics in Medicine},
number = {19},
pages = {3388--3414},
title = {{A tutorial on propensity score estimation for multiple treatments using generalized boosted models}},
volume = {32},
year = {2013}
}

@article{yu2025balancing,
  title={Balancing weights for causal inference in observational factorial studies},
  author={Yu, Ruoqi and Ding, Peng},
  journal={Journal of the American Statistical Association},
  year={2025}
}

@article{slud1994analysis,
  title={Analysis of factorial survival experiments},
  author={Slud, Eric V},
  journal={Biometrics},
  pages={25--38},
  year={1994},
  publisher={JSTOR}
}

@article{cox1972regression,
  title={Regression models and life-tables},
  author={Cox, David R},
  journal={Journal of the Royal Statistical Society: Series B (Methodological)},
  volume={34},
  number={2},
  pages={187--202},
  year={1972},
  publisher={Wiley Online Library}
}

@article{cox1975partial,
  title={Partial likelihood},
  author={Cox, David R},
  journal={Biometrika},
  volume={62},
  number={2},
  pages={269--276},
  year={1975},
  publisher={Oxford University Press}
}

@article{kahan2023reporting,
  title={Reporting of factorial randomized trials: extension of the CONSORT 2010 statement},
  author={Kahan, Brennan C and Hall, Sophie S and Beller, Elaine M and Birchenall, Megan and Chan, An-Wen and Elbourne, Diana and Little, Paul and Fletcher, John and Golub, Robert M and Goulao, Beatriz and others},
  journal={Jama},
  volume={330},
  number={21},
  pages={2106--2114},
  year={2023},
  publisher={American Medical Association}
}

@article{simms2024effect,
  title={Effect of combination treatment with glucagon-like peptide-1 receptor agonists and sodium-glucose cotransporter-2 inhibitors on incidence of cardiovascular and serious renal events: population based cohort study},
  author={Simms-Williams, Nikita and Treves, Nir and Yin, Hui and Lu, Sally and Yu, Oriana and Pradhan, Richeek and Renoux, Christel and Suissa, Samy and Azoulay, Laurent},
  journal={bmj},
  volume={385},
  year={2024},
  publisher={British Medical Journal Publishing Group}
}

@article{lin2016simultaneous,
  title={Simultaneous inference on treatment effects in survival studies with factorial designs},
  author={Lin, Dan-Yu and Gong, Jianjian and Gallo, Paul and Bunn, Paul H and Couper, David},
  journal={Biometrics},
  volume={72},
  number={4},
  pages={1078--1085},
  year={2016},
  publisher={Wiley Online Library}
}

@article{binder1992fitting,
  title={Fitting Cox's proportional hazards models from survey data},
  author={Binder, David A},
  journal={Biometrika},
  volume={79},
  number={1},
  pages={139--147},
  year={1992},
  publisher={Oxford University Press}
}

@article{kalapura2026real,
title={Real-World Effectiveness of GLP-1 Receptor Agonists on Clinical Outcomes in Patients with Heart Failure with Preserved Ejection Fraction 
},
  author={Cheryl Kalapura and Jiwon Shin and Robert J. Mentz and Stephen J. Greene and Brian Mac Grory and Fan Li and Jay B. Lusk and Emily C. O’Brien},
  journal={Circulation},
  volume={152},
  number={Suppl 3},
  pages={A4370283-A4370283},
  year={2025},
  publisher={}
}

@techreport{FDAcovadj2023,
  author = {{U.S. Food and Drug Administration}},
  title = {Adjusting for Covariates in Randomized Clinical Trials for Drugs and Biological Products},
  year = {2023},
  month = {May},
  type = {GUIDANCE DOCUMENT},
  number = {FDA-2019-D-0934},
  institution = {{U.S. Food and Drug Administration}},
  address = {Silver Spring, MD}
}

@article{gabriel2024propensity,
  title={Propensity weighting plus adjustment in proportional hazards model is not doubly robust},
  author={Gabriel, Erin E and Sachs, Michael C and Waernbaum, Ingeborg and Goetghebeur, Els and Blanche, Paul F and Vansteelandt, Stijn and Sj{\"o}lander, Arvid and Scheike, Thomas},
  journal={Biometrics},
  volume={80},
  number={3},
  pages={ujae069},
  year={2024},
  publisher={Oxford University Press}
}

@article{fay2024causal,
  title={Causal interpretation of the hazard ratio in randomized clinical trials},
  author={Fay, Michael P and Li, Fan},
  journal={Clinical Trials},
  volume={21},
  number={5},
  pages={623--635},
  year={2024},
  publisher={SAGE Publications Sage UK: London, England}
}

@article{rosenbaum1987model,
  title={Model-based direct adjustment},
  author={Rosenbaum, Paul R},
  journal={Journal of the American statistical Association},
  volume={82},
  number={398},
  pages={387--394},
  year={1987},
  publisher={Taylor \& Francis}
}

@article{hernan2000marginal,
  title={Marginal structural models to estimate the causal effect of zidovudine on the survival of HIV-positive men},
  author={Hern{\'a}n, Miguel {\'A}ngel and Brumback, Babette and Robins, James M},
  journal={Epidemiology},
  volume={11},
  number={5},
  pages={561--570},
  year={2000},
  publisher={LWW}
}

@Manual{PSweight,
    title = {PSweight: Propensity Score Weighting for Causal Inference with Observational Studies and Randomized Trials},
    author = {Tianhui Zhou and Guangyu Tong and Fan Li and Laine Thomas and Fan Li and Yukang Zeng},
    year = {2025},
    note = {R package version 2.1.2},
    url = {https://cran.r-project.org/web/packages/PSweight/},
  }

@Manual{twang,
    title = {twang: Toolkit for Weighting and Analysis of Nonequivalent Groups},
    author = {Matthew Cefalu and Greg Ridgeway and Dan McCaffrey and Andrew Morral and Beth Ann Griffin and Lane Burgette},
    year = {2025},
    note = {R package version 2.6.2},
    url = {https://cran.r-project.org/web/packages/twang/},
  }

@Manual{PSsurvival,
    title = {PSsurvival:Propensity Score Methods for Survival Analysis},
    author = {Chengxin Yang and Fan Li and Chao Cheng and Fan Li},
    year = {2026},
    note = {R package version 0.2.0},
    url = {https://cran.r-project.org/web/packages/PSsurvival/},
  }

@article{xu2000AverageHazard,
    author = {Xu, Ronghui and O’Quigley, John},
    title = {Estimating average regression effect under non-proportional
  hazards },
    journal = {Biostatistics},
    volume = {1},
    number = {4},
    pages = {423-439},
    year = {2000},
    month = {12},
    issn = {1465-4644}
}

@article{martinussen2022causality,
  title={Causality and the Cox regression model},
  author={Martinussen, Torben},
  journal={Annual Review of Statistics and Its Application},
  volume={9},
  number={1},
  pages={249--259},
  year={2022},
  publisher={Annual Reviews}
}

@article{Luo2024drMHR,
    author = {Luo, Jiyu and Rava, Denise and Bradic, Jelena and Xu, Ronghui},
    title = {Doubly robust estimation under a possibly misspecified marginal structural Cox model},
    journal = {Biometrika},
    volume = {112},
    number = {1},
    pages = {asae065},
    year = {2025},
    month = {02},
    issn = {1464-3510}
}

@article{cui2018application,
  title={Application and comparison of generalized propensity score matching versus pairwise propensity score matching},
  author={Cui, Zhanglin L and Hess, Lisa M and Goodloe, Robert and Faries, Doug},
  journal={Journal of comparative effectiveness research},
  volume={7},
  number={09},
  pages={923--934},
  year={2018},
  publisher={Future Medicine Ltd London, UK}
}

@article{sugihara2010gps,
  title={Survival analysis using inverse probability of treatment weighted methods based on the generalized propensity score},
  author={Sugihara, Masahiro},
  journal={Pharmaceutical Statistics: The Journal of Applied Statistics in the Pharmaceutical Industry},
  volume={9},
  number={1},
  pages={21--34},
  year={2010},
  publisher={Wiley Online Library}
}

@article{kuitunen2022review,
  title={Survival analysis and cox proportional hazards model reporting in pediatric leukemia studies—a systematic review},
  author={Kuitunen, Ilari and Nikkil{\"a}, Atte and Ponkilainen, Ville T and Uimonen, Mikko M and Lohi, Olli},
  journal={SN Comprehensive Clinical Medicine},
  volume={5},
  number={1},
  pages={24},
  year={2022},
  publisher={Springer}
}

@article{rulli2018assessment,
  title={Assessment of proportional hazard assumption in aggregate data: a systematic review on statistical methodology in clinical trials using time-to-event endpoint},
  author={Rulli, Eliana and Ghilotti, Francesca and Biagioli, Elena and Porcu, Luca and Marabese, Mirko and D’Incalci, Maurizio and Bellocco, Rino and Torri, Valter},
  journal={British journal of cancer},
  volume={119},
  number={12},
  pages={1456--1463},
  year={2018},
  publisher={Nature Publishing Group UK London}
}

@article{andersen1982cox,
    title = {Cox's Regression Model for Counting Processes: A Large Sample Study},
    author = {P. K. Andersen and R. D. Gill},
    ISSN = {00905364, 21688966},
    journal = {The Annals of Statistics},
    number = {4},
    pages = {1100--1120},
    publisher = {Institute of Mathematical Statistics},
    volume = {10},
    year = {1982}
}

@article{Lin1991goodnessfit,
author = {D. Y. Lin},
title = {Goodness-of-Fit Analysis for the Cox Regression Model Based on a Class of Parameter Estimators},
journal = {Journal of the American Statistical Association},
volume = {86},
number = {415},
pages = {725--728},
year = {1991},
publisher = {Taylor \& Francis},
doi = {10.1080/01621459.1991.10475101}
}

@article{schoenfeld1982residuals,
    author = {SCHOENFELD, DAVID},
    title = {Partial residuals for the proportional hazards regression model},
    journal = {Biometrika},
    volume = {69},
    number = {1},
    pages = {239-241},
    year = {1982},
    month = {04},
    issn = {0006-3444}
}

@article{grambsch1994testprophazard,
    author = {Grambsch, PM and Therneau, TM},
    title = {Proportional hazards tests and diagnostics based on weighted residuals},
    journal = {Biometrika},
    volume = {81},
    number = {3},
    pages = {515-526},
    year = {1994},
    month = {09},
    issn = {0006-3444}
}

@article{Austin2013estimate,
author = {Austin, Peter C.},
title = {The performance of different propensity score methods for estimating marginal hazard ratios},
journal = {Statistics in Medicine},
volume = {32},
number = {16},
pages = {2837-2849},
year = {2013}
}

@article{Beyersmann2025hazards,
author = {Beyersmann, Jan and Schmoor, Claudia and Schumacher, Martin},
title = {Hazards Constitute Key Quantities for Analyzing, Interpreting and Understanding Time-to-Event Data},
journal = {Biometrical Journal},
volume = {67},
number = {3},
pages = {e70057},
year = {2025}
}

@article{Struthers1986misspecify,
    author = {STRUTHERS, C. A. and KALBFLEISCH, J. D.},
    title = {Misspecified proportional hazard models},
    journal = {Biometrika},
    volume = {73},
    number = {2},
    pages = {363-369},
    year = {1986},
    month = {08},
    issn = {0006-3444}
}

@article{greenland1996absence,
  title={Absence of confounding does not correspond to collapsibility of the rate ratio or rate difference},
  author={Greenland, Sander},
  journal={Epidemiology},
  volume={7},
  number={5},
  pages={498--501},
  year={1996},
  publisher={Lippincott Williams \& Wilkins},
  pmid={8862980}
}

@article{martinussen2013collapsibility,
  author  = {Martinussen, Torben and Vansteelandt, Stijn},
  title   = {On collapsibility and confounding bias in {C}ox and {A}alen regression models},
  journal = {Lifetime Data Analysis},
  year    = {2013},
  volume  = {19},
  number  = {3},
  pages   = {279--296},
  doi     = {10.1007/s10985-013-9242-z}
}

\newpage
\spacingset{1.25}

\appendix
\titleformat{\section}[block]
  {\normalfont\Large\bfseries\centering}
  {Appendix \thesection.}{0.5em}{}
\titleformat{\subsection}[block]
  {\normalfont\large\bfseries\centering}
  {\thesubsection}{0.5em}{}
\titleformat{\subsubsection}[block]
  {\normalfont\normalsize\bfseries}
  {\thesubsubsection}{0.5em}{}

\section{Proof of Theorem 1}

{Throughout the appendix we follow the notation of Section \ref{sec:model}. We adopt the weighted partial likelihood, \cite{cox1975partial, binder1992fitting} and follow the standard arguments on regularity conditions, counting process, and asymptotic linearization. \cite{andersen1982cox, lin1989robust} We stack the estimating equations of Cox model parameters and propensity scores following Shu et al., \cite{shu2021variance} extending their setting to multiple treatments.}

\subsection{Setup and notations}

{We follow the definitions in Section \ref{sec:model}. The observed data are $n$ independent and identically distributed copies of $O_i \!=\! (U_i,\delta_i,Z_i,X_i)$, with $U_i \!=\! \min(T_i,C_i)$, $\delta_i \!=\! \one\{T_i\le C_i\}$, $T_i \!=\! T_i(Z_i)$ and $C_i \!=\! C_i(Z_i)$. Let $e_z(x)\!=\! \Pr(Z\!=\!z \mid X\!=\!x)$ denote the propensity to treatment $z$. Let $f(x)$ be the density of $X$ in the study population, and let} $\lambda_z(t \mid x)$ {and} $S_z(t\mid x)$ {be} the {conditional} hazard and survival {functions of $T(z)$ given $X=x$},
\begin{align*}
\lambda_{{z}}(t\mid {x})=\lim_{\Delta t\downarrow 0}
\frac{\Pr\{t<T({z})\le t+\Delta t\mid T({z})\ge t, {X=x} \}}{\Delta t}, \quad
S_{{z}}(t\mid {x} )=\Pr\{T({z})\ge t\mid {X=x} \}.
\end{align*}

{A target population is defined by a balancing weight $w_z(x) \!=\! h(x)/e_z(x)$, where $w_z(x)$ is the weight to treatment $z$ given covariates and $h$ is a pre-specified tilting function. The covariate density of that population is $f^w(x) \propto f(x)h(x)$. The IPW, ATT (targeting at a specific treatment $z^\prime$), and OW of Section \ref{sec:Cox-multi} correspond to $h(x) \!=\! 1$, $h(x) \!=\! e_{z^\prime}(x)$, and $h(x) \!=\! \{\sum_{z\in\mathbb{Z}}1/e_z(x)\}^{-1}$, where the latter two depend on the propensity scores.}

{For the target population of weight $w$, the (marginal) counterfactual survival and hazard functions of treatment $z$ are $S_z^w(t)=E\{h(X)S_z(t\mid X)\}/E\{h(X)\}$ and}
\begin{align*}
{\lambda_z^w(t)} = \frac{E\!\left[h(X)\lambda_{{z}}(t\mid X)\,S_{{z}}(t\mid X)\right]}{E\!\left[h(X)S_{{z}}(t\mid X)\right]},
\end{align*}
{where $h$ is the tilting function corresponding to $w$ and the log MHR of treatment $z$ is $\tau_z^w(t)=\log\{\lambda_z^w(t)/\lambda_0^w(t)\}$.}
We assume {proportional hazards} in the target population, {$\tau_z^w(t)\equiv\tau_z^w$ for all $t$, so that $\tau_z^w$ coincides with the coefficient of treatment $z$ in the following} marginal Cox model in the target population:
\begin{equation}
{\lambda_z^w(t)}={\lambda_0^w(t)} \exp(\tau_{{z}}^w),
\qquad {z=1,\ldots,J},
\label{eq:target_cox_app}
\end{equation}
{and further denote $\tau^w=(\tau_1^w, \ldots, \tau_J^w)$. As in Section \ref{sec:model}}, we {assume a multinomial logistic propensity score model},
\begin{align*}
e_0(x;\gamma)=\frac{1}{1+\sum_{{j=1}}^J\exp(\gamma_j^\prime x)},\quad
e_{{z}}(x;\gamma)=\frac{\exp(\gamma_z^\prime x)}
{1+\sum_{{j=1}}^J\exp(\gamma_j^\prime x)}, \quad {z}=1,\ldots,J,
\end{align*}
where {$\gamma_z$ is the covariate coefficient for treatment $z$ and denote} \(\gamma=(\gamma_1, \ldots, \gamma_J)^\prime\). 
{Therefore,} the weight for unit $i$ is {$w_i(\gamma)=w_{i,Z_i}(\gamma) = h(X_i;\gamma)/e_{Z_i}(X_i;\gamma)$, where $Z_i$ is the observed treatment.} 

We maintain {Assumptions (A1)-(A4) of Section \ref{sec:Cox-multi}}. 
{We further assume (A5) the propensity score model is correctly specified and consistently estimated, that is, $e_z(x;\gamma^*)=e_z(x)$ for some $\gamma^*$ and all $z$, and its MLE $\hat{\gamma}\xrightarrow{p} \gamma^*$; and (A6) $C(z) \leq T_0$ almost surely for some finite $T_0 < \infty$.} 
The weighted partial likelihood estimator $\hat{\tau}^w$ for $\tau^w$ and the MLE $\hat{\gamma}$ for $\gamma$ are obtained by jointly solving the following stacked estimating equation,
\begin{equation}\label{eq:joint}
\sum_{i=1}^n \phi_i(\tau^w,\gamma)=\sum_{i=1}^n
\begin{pmatrix}
\psi_i(\tau^w,\gamma)\\
\pi_i(\gamma)
\end{pmatrix}
= \mathbf 0,
\end{equation}
where $\psi_i(\tau^w,\gamma)=w_i(\gamma)\delta_i\{\mathbf{D}_i-\bar{\mathbf{D}}({U_i};\tau^w,\gamma)\}$ and $\pi_i(\gamma)=\big(\mathbf D_i-\mathbf e(X_i;\gamma)\big) \!\otimes\! X_i$, {whose components are defined in Section \ref{sec:var-robust}}. 
{Notice that (A5) implies the true weights $w_i\!=\!w_i(\gamma^*)$. At fixed weight, the partial likelihood estimating equation for $\tau^w$ is thus $\sum_{i=1}^n \psi_i(\tau^w,\gamma^*)\!=\!\mathbf 0$ under the true weights and $\sum_{i=1}^n \psi_i(\tau^w, \hat\gamma)\!=\!\mathbf 0$ under the estimated weights. We denote their solutions by $\hat{\tau}^w(\gamma^*)$ and $\hat{\tau}^w(\hat\gamma)$, the latter coinciding with $\hat{\tau}^w$ in the joint equation \eqref{eq:joint}.}

{Finally, alongside the} at-risk processes $Y_i(t)\!=\! \one\{{U_i} \!\geq\! t\}$ {of Section \ref{sec:var-robust}, we define} the counting process $N_i(t)\!=\! \one\{{U_i} \!\leq\! t, \delta_i\!=\!1\}$; further define the weighted empirical cumulative event process {and its population mean}
\begin{equation*}
N^w(t; \gamma)=\sum_{i=1}^n w_i(\gamma) N_i(t),\quad
\hat F_n(t; \gamma)= \frac{N^w(t;\gamma)}{n}, \quad {F(t;\gamma)=E[w_i(\gamma) N_i(t)]}
\end{equation*}

The proof proceeds in three steps.
First, we derive key weighted population moment identities.
Second, we show that the population {partial likelihood} score vanishes at the {true log MHRs} $\tau^w$, {which establishes the consistency of $\hat{\tau}^w(\gamma^*)$ and then that of $\hat{\tau}^w(\hat\gamma)$}.
Third, we derive an asymptotic linear representation of {$\hat{\tau}^w(\hat\gamma)$} and obtain the robust estimator {for its variance}.

\subsection{Key identities for weighted population moments}
By Assumption {(A4)}, for each {treatment $z$} and $t\in[0,T_0]$,
\begin{equation*}
\Pr\{C({z})\ge t\mid {X=x}, Z={z} \}=\Pr\{C({z})\ge t\mid Z={z} \}=:G_{{z}}(t).
\end{equation*}

For $t\in[0,T_0]$ and each {$z$}, define
\[
R_{{z}}(t\mid X):=
E\!\left[\one\{T({z})\ge t\}\one\{C({z})\ge T({z})\}\mid X,Z={z}\right].
\]
In particular $R_{{z}}(0\mid X)=E[\one\{C({z})\ge T({z})\}\mid X,Z={z}]$.

\begin{lemma}[Integral representation of $R_{{z}}(t\mid X)$]
\label{lem:R_integral}
Under {(A1)-(A6)},
\[
R_{{z}}(t\mid X)=\int_t^{T_0}\lambda_{{z}}(s\mid X)S_{{z}}(s\mid X)\,G_{{z}}(s)\,ds.
\]
\end{lemma}

\begin{proof}
Fix ${z}$ and condition on $X$ and $Z={z}$.
By the law of total probability and {(A4)},
\begin{align*}
R_{{z}}(t\mid X)
&=E\!\left[\one\{T({z})\ge t\}\one\{C({z})\ge T({z})\}\mid X,Z={z}\right]\\
&=\int_t^{T_0}\Pr\{C({z})\ge s\mid X,Z={z},T({z})=s\}\,f_{T({z})\mid X,Z={z}}(s)\,ds\\
&=\int_t^{T_0}\Pr\{C({z})\ge s\mid Z={z}\}\,f_{T({z})\mid X}(s)\,ds\\
&=\int_t^{T_0}G_{{z}}(s)\,\lambda_{{z}}(s\mid X)S_{{z}}(s\mid X)\,ds,
\end{align*}
where we used {(A2)} to replace $f_{T({z})\mid X,Z={z}}$ by $f_{T({z})\mid X}$ and the identity
$f_{T({z})\mid X}(s)=\lambda_{{z}}(s\mid X)S_{{z}}(s\mid X)$.
\end{proof}

\begin{lemma}[Weighted population moments]
\label{lem:weighted_moments}
Let $u:\{0,\ldots,J\}\to\mathbb{R}$ be any function.
Under (A1)--{(A6)}, for all $t\in[0,T_0]$ and ${z}\in\{0,\ldots,J\}$,
\begin{align}
E\!\left[w\,\delta\,\one\{Z={z}\}\right]
&=E\!\left[h(X)\,R_{{z}}(0\mid X)\right],
\label{eq:moment1}\\[3pt]
E\!\left[w\,u(Z)\,\one\{{U}\ge t\}\right]
&=\sum_{j=0}^J u(j)\,E\!\left[h(X)\,S_j(t\mid X)\,G_j(t)\right],
\label{eq:moment2}\\[3pt]
E\!\left[w\,\delta\,\one\{{U}\ge t\}\right]
&=\sum_{j=0}^J E\!\left[h(X)\,R_j(t\mid X)\right].
\label{eq:moment3}
\end{align}
\end{lemma}

\begin{proof}
\emph{Proof of \eqref{eq:moment1}.}
Using the law of total expectation,
{
\begin{align*}
E\!\left[w\,\delta\,\one\{Z=z\}\right]
&=E\!\left[E\!\left\{w\,\delta\,\one\{Z=z\}\mid X\right\}\right]\\
&=E\!\left[\frac{h(X)}{e_{z}(X)}\,E\!\left\{\delta\,\one\{Z=z\}\mid X\right\}\right]\\
&=E\!\left[\frac{h(X)}{e_{z}(X)}\,\Pr(Z=z\mid X)\,E\!\left\{\delta\mid X,Z=z\right\}\right]\\
&=E\!\left[h(X)\,E\!\left\{\delta\mid X,Z=z\right\}\right].
\end{align*}
}
By consistency, {on $\{Z=z\}$ we have $\delta=\one\{C(z)\ge T(z)\}$}, hence
\begin{equation*}
    E\!\left\{\delta\mid X,Z=z\right\} = E\!\left[\one\{C(z)\ge T(z)\}\mid X,Z=z\right] =R_{z}(0\mid X),
\end{equation*}
which gives \eqref{eq:moment1}.

\vspace{12pt}
\emph{Proof of \eqref{eq:moment2}.}
Using $\sum_{j=0}^J\one\{Z=j\}=1$ and total expectation,
\begin{align*}
E\!\left[w\,u(Z)\,\one\{{U}\ge t\}\right]
&=E\!\left[E\!\left\{w\,u(Z)\,\one\{{U}\ge t\}\mid X\right\}\right]\\
&=E\!\left[\sum_{j=0}^J \frac{h(X)}{e_j(X)}u(j)\,
E\!\left\{\one\{Z=j\}\one\{{U}\ge t\}\mid X\right\}\right]\\
&=E\!\left[\sum_{j=0}^J \frac{h(X)}{e_j(X)}u(j)\,
\Pr(Z=j\mid X)\,E\!\left\{\one\{{U}\ge t\}\mid X,Z=j\right\}\right]\\
&=\sum_{j=0}^J u(j)\,E\!\left[h(X)\,E\!\left\{\one\{{U}\ge t\}\mid X,Z=j\right\}\right].
\end{align*}
By consistency, {$\one\{U\ge t\}=\one\{T(z)\ge t\}\one\{C(z)\ge t\}$ on $\{Z=z\}$ for  $z=0,\ldots,J$}, and by {(A4)},
{
\begin{align*}
E\!\left\{\one\{{U}\ge t\}\mid X,Z=z\right\}
&=E\!\left[\one\{T(z)\ge t\}\one\{C(z)\ge t\}\mid X,Z=z\right]\\
&=\Pr\{T(z)\ge t\mid X\}\Pr\{C(z)\ge t\mid Z=z\}\\
&=S_z(t\mid X)\,G_z(t).
\end{align*}
}
Substituting yields \eqref{eq:moment2}.

\vspace{12pt}
\emph{Proof of \eqref{eq:moment3}.}
Similarly,
\begin{align*}
E\!\left[w\,\delta\,\one\{{U}\ge t\}\right]
&=E\!\left[E\!\left\{w\,\delta\,\one\{{U}\ge t\}\mid X\right\}\right]\\
&=E\!\left[\sum_{j=0}^J \frac{h(X)}{e_j(X)}\,
E\!\left\{\delta\,\one\{{U}\ge t\}\one\{Z=j\}\mid X\right\}\right]\\
&=\sum_{j=0}^J E\!\left[h(X)\,E\!\left\{\delta\,\one\{{U}\ge t\}\mid X,Z=j\right\}\right],
\end{align*}
and {for $z=0,\ldots,J$, 
\[
E\!\left\{\delta\,\one\{{U}\ge t\}\mid X,Z=z\right\}
=
E\!\left[\one\{T(z)\ge t\}\one\{C(z)\ge T(z)\}\mid X,Z=z\right]
=R_z(t\mid X).
\]
}
Substituting gives \eqref{eq:moment3}.
\end{proof}

\subsection{{Consistency of $\hat\tau^w(\hat\gamma)$ for $\tau^w$}}

Define the weighted partial log-likelihood of the Cox model constructed with the true weights by {
\begin{align}
\mathcal{G}_n(\tau^w)
&=\frac{1}{n}\sum_{i=1}^{n} w_i \delta_i\Big(\sum_{j=1}^J \one\{Z_i=j\}\tau_j^w\Big)
-\frac{1}{n}\sum_{i=1}^{n} w_i \delta_i
\log\Big(\frac{1}{n}\sum_{l\in \mathcal{R}_i} w_l\exp\{\sum_{j=1}^J \one\{Z_l=j\}\tau_j^w\}\Big),
\label{eq:Fn_def}
\end{align} }
where {$\mathcal{R}_i \! = \! \{l \! : \! U_i\leq U_l \mid l \! = \! 1, \dots, n \}$} denotes the risk set {at $U_i$}.

Rewriting the sums over event times as Stieltjes integrals,
\eqref{eq:Fn_def} can be equivalently written as
{
\begin{align*}
\mathcal{G}_n(\tau^w)
&=\sum_{j=1}^J\tau_j^w\cdot \frac{1}{n}\sum_{i=1}^{n} w_i \delta_i \one\{Z_i=j\}-\int_0^{T_0} -d\Big(\frac{1}{n}\sum_{i=1}^{n} w_i \delta_i \one\{{U}_i\ge s\}\Big)\,
\log\Big(\frac{1}{n}\sum_{l=1}^{n} w_l e^{\sum_{j=1}^J \one\{Z_l=j\}\tau_j^w}\one\{{U}_l\ge s\}\Big),
\end{align*}
and the estimating equation $\sum_{i=1}^n \psi_i(\tau^w,\gamma^*)\!=\!\mathbf 0$ for $\tau^w$ under the true weights is equivalent to $\nabla_{\tau^w} \mathcal{G}_n(\tau^w)=\mathbf{0}$}.

\medskip
By the law of large numbers and Lemma~\ref{lem:weighted_moments},
as $n\to\infty$,
\begin{align*}
\mathcal{G}_n(\tau^w)\ \longrightarrow\ 
\mathcal{G}(\tau^w)
&:= \sum_{j=1}^J\tau_j^w\cdot E\!\left[h(X)R_j(0\mid X)\right] \\
& \qquad -\int_0^{T_0} -d\Big(\sum_{{j}=0}^J E[h(X)R_{{j}}(s\mid X)]\Big)\,
\log\Big(\sum_{{j}=0}^J e^{\tau^w_{{j}}}\,E[h(X)S_{{j}}(s\mid X)G_{{j}}(s)]\Big).
\end{align*}

\begin{lemma}[Population score vanishes at $\tau^w$]
\label{lem:score_zero}
Under (A1)--{(A6)} and the target marginal Cox model \eqref{eq:target_cox_app},{
$\nabla_{\tau^w} \mathcal{G}(\tau^w)$ equals $\mathbf{0}$.}
\end{lemma}
\begin{proof}
{Fix $z\in\{1,\ldots,J\}$. Differentiating $\mathcal{G}(\tau^w)$ with respect to $\tau^w_{z}$ yields
\begin{align}
\frac{\partial \mathcal{G}(\tau^w)}{\partial \tau^w_z}
&=E[h(X)R_{z}(0\mid X)]
-\int_0^{T_0} -d\Big(\sum_{j=0}^J E[h(X)R_{j}(s\mid X)]\Big)\,
\frac{e^{\tau^w_z}\,E[h(X)S_{z}(s\mid X)G_{z}(s)]}
{\sum_{j=0}^J e^{\tau^w_{j}}\,E[h(X)S_j(s\mid X)G_{j}(s)]}.
\label{eq:pop_score}
\end{align}
It remains to show that the right-hand side of \eqref{eq:pop_score} equals $0$.}

First note from Lemma~\ref{lem:R_integral}, for $j = 0, \ldots, J$,
\begin{align}
\frac{d}{ds}E[h(X)R_{j}(s\mid X)]
&=-E\!\left[h(X)\lambda_{j}(s\mid X)S_{j}(s\mid X)\right]\,G_{j}(s),
\label{eq:dRh}
\end{align}
where $G_{j}(s)$ does not depend on $X$ by {(A4)}.

Next, under marginal Cox model
\eqref{eq:target_cox_app},
{
\[
\lambda_j^{w}(s)
=\frac{E[h(X)\lambda_{j}(s\mid X)S_{j}(s\mid X)]}{E[h(X)S_{j}(s\mid X)]}
=\lambda_0^{w}(s)e^{\tau_{j}^w},
\]
which implies
\begin{equation}
E[h(X)\lambda_{j}(s\mid X)S_{j}(s\mid X)]
=\lambda_0^{w}(s)e^{\tau_{j}^w}\,E[h(X)S_{j}(s\mid X)].
\label{eq:key_identity}
\end{equation}
}
Substituting \eqref{eq:key_identity} into \eqref{eq:dRh} gives
{
\begin{align}
\frac{d}{ds}E[h(X)R_{j}(s\mid X)]
=-\lambda_0^{w}(s)e^{\tau_{j}^w}\,E[h(X)S_{j}(s\mid X)G_{j}(s)].
\label{eq:dRh_tauh}
\end{align}}
Summing \eqref{eq:dRh_tauh} over ${{j}}=0,\ldots,J$ yields
{
\[
-d\Big(\sum_{j=0}^J E[h(X)R_{j}(s\mid X)]\Big)
=\lambda_0^{w}(s)\Big(\sum_{j=0}^J e^{\tau_{j}^w}E[h(X)S_{j}(s\mid X)G_{j}(s)]\Big)\,ds.
\]}
Plugging this into the integral term of \eqref{eq:pop_score} gives{
\begin{align*}
&\int_0^{T_0} -d\Big(\sum_{j=0}^J E[h(X)R_j(s\mid X)]\Big)\,
\frac{e^{\tau_z^w}\,E[h(X)S_z(s\mid X)G_z(s)]}
{\sum_{j=0}^J e^{\tau_j^w}\,E[h(X)S_j(s\mid X)G_j(s)]}\\
&\qquad =\int_0^{T_0}\lambda_0^{w}(s)e^{\tau_z^w}E[h(X)S_z(s\mid X)G_z(s)]\,ds.
\end{align*}}

Finally, applying \eqref{eq:key_identity} with {$j=z$} and Lemma~\ref{lem:R_integral} yields
{
\begin{align*}
    \int_0^{T_0}\lambda_0^{w}(s)e^{\tau_z^w}E[h(X)S_z(s\mid X)G_z(s)]\,ds =&\int_0^{T_0}E[h(X)\lambda_z(s\mid X)S_z(s\mid X)]\,G_z(s)\,ds \\
    =&E[h(X)R_z(0\mid X)].
\end{align*}}
Therefore, \eqref{eq:pop_score} equals $0$ at $\tau^w$, which proves the claim.
\end{proof}

Lemma \ref{lem:score_zero} implies that
\[
{\hat\tau^w(\gamma^*)} \xrightarrow{p} \tau^w.
\]

We now show that replacing the true weights $w_i=w_i(\gamma^*)$ by the estimated weights
$w_i(\hat\gamma)$ does not affect consistency.
By Assumption {(A5)}, $\hat\gamma \xrightarrow{p} \gamma^*$, and hence
\[
\frac{1}{n}\sum_{i=1}^n
| w_i(\hat\gamma)-w_i(\gamma^*) |
\xrightarrow{p} 0.
\]
Moreover, the estimating function {$\psi_i(\tau^w,\gamma)$} is continuous in
$\gamma$ and dominated by an integrable envelope under the overlap condition. It follows that{
\[
\sup_{\tau^w}
\left\|
\frac{1}{n}\sum_{i=1}^n \psi_i(\tau^w,\hat\gamma)
-
\frac{1}{n}\sum_{i=1}^n \psi_i(\tau^w,\gamma^*)
\right\|
=o_p(1).
\]}

Consequently,{
\[
\hat\tau^w(\hat\gamma) - \hat\tau^w(\gamma^*) = o_p(1),
\]
and therefore
\[
\hat\tau^w =\hat\tau^w(\hat\gamma)\xrightarrow{p} \tau^w.
\]}

\subsection{Asymptotic linear representation}

The following expansion is the usual asymptotic linear representation of the Cox partial likelihood
score, modified by subject-level weights.\citep{lin1989robust}\citep{binder1992fitting}

Let{
\[
s^{(0)}(t;\tau^w, \gamma)=E[w_i(\gamma) Y_i(t)e^{\eta_i(\tau^w)}],\qquad
\mathbf{s}^{(1)}(t;\tau^w, \gamma)=E[w_i(\gamma) Y_i(t)\mathbf{D}_i e^{\eta_i(\tau^w)}],
\]
and define
\[
\bar{\mathbf d}(t;\tau^w, \gamma)=\mathbf{s}^{(1)}(t;\tau^w, \gamma)/s^{(0)}(t;\tau^w, \gamma).
\]
}

Assume that, uniformly in $t$,{
\[
\sup_t \|S^{(r)}(t;\tau^w, \gamma)-s^{(r)}(t;\tau^w, \gamma)\|\xrightarrow{p}0,
\qquad
\hat F_n(t; \gamma)\xrightarrow{p}F(t; \gamma).
\]}
Then
\begin{align*}
    n^{-\frac{1}{2}}\sum_{i=1}^n \psi_i(\tau^w,\gamma)&=n^{-\frac{1}{2}}\sum_{i=1}^n w_i(\gamma)\delta_i\{
    \mathbf{D}_i-{\bar{\mathbf{D}}(U_i;\tau^w,\gamma)}\}
    =n^{-\frac{1}{2}}\sum_{i=1}^n \int  w_i(\gamma)\{\mathbf{D}_i-\bar{\mathbf{D}}(t;\tau^w,\gamma)\}\, dN_i(t)\\
    &=n^{-\frac{1}{2}}\sum_{i=1}^n \int
    w_i(\gamma) \mathbf{D}_i\ dN_i(t)-n^{-\frac{1}{2}} \int \bar{\mathbf{D}}(t;\tau^w,\gamma)d(\sum_i w_i(\gamma) N_i(t))\\
    &=n^{-\frac{1}{2}}\sum_{i=1}^n \int
    w_i(\gamma) \mathbf{D}_i\ dN_i(t)-n^{\frac{1}{2}} \int \bar{\mathbf{D}}(t;\tau^w,\gamma)d\hat{F}_n(t; \gamma).
\end{align*}

A first-order Taylor expansion of $\bar{\mathbf{D}}(t;\tau^w,\gamma)$ around
$\bar{\mathbf d}(t;\tau^w,\gamma)$ yields {
\begin{align*}
\bar{\mathbf{D}}(t;\tau^w,\gamma)=&\bar{\mathbf{d}}(t;\tau^w,\gamma)+
\frac{\mathbf{S}^{(1)}(t;\tau^w,\gamma) -\mathbf{s}^{(1)}(t;\tau^w,\gamma)}{s^{(0)}(t;\tau^w,\gamma)}
-\frac{\bar{\mathbf{d}}(t;\tau^w,\gamma)}{s^{(0)}(t;\tau^w,\gamma)}
\big\{S^{(0)}(t;\tau^w,\gamma)-s^{(0)}(t;\tau^w,\gamma)\big\} +o_p(n^{-1/2})\\
=&\bar{\mathbf{d}}(t;\tau^w,\gamma)+
\frac{1}{s^{(0)}(t;\tau^w,\gamma)}[\mathbf{S}^{(1)}(t;\tau^w,\gamma)
-\frac{\mathbf{s}^{(1)}(t;\tau^w,\gamma)}{s^{(0)}(t;\tau^w,\gamma)}
S^{(0)}(t;\tau^w,\gamma)]+o_p(n^{-1/2}).
\end{align*}
}

Substituting this expansion and rearranging terms gives{
\begin{align*}
    &n^{-\frac{1}{2}}\sum_{i=1}^n \psi_i(\tau^w,\gamma^*)
    =n^{-\frac{1}{2}}\sum_{i=1}^n \int
    w_i(\gamma^*) \mathbf{D}_i\ dN_i(t)-n^{\frac{1}{2}} \int \bar{\mathbf{D}}(t;\tau^w,\gamma^*)d\hat{F}_n(t; \gamma^*)\\
    =&n^{-\frac{1}{2}}\sum_{i=1}^n \int
    w_i(\gamma^*) \mathbf{D}_i\ dN_i(t)\\
    &-n^{\frac{1}{2}} \int \{\bar{\mathbf{d}}(t;\tau^w,\gamma^*)+\frac{1}{s^{(0)}(t;\tau^w,\gamma^*)}[\mathbf{S}^{(1)}(t;\tau^w,\gamma^*)-\frac{\mathbf{s}^{(1)}(t;\tau^w,\gamma^*)}{s^{(0)}(t;\tau^w,\gamma^*)}S^{(0)}(t;\tau^w,\gamma^*)]\}d\hat{F}_n(t; \gamma^*)+o_p(1)\\
    =&n^{-\frac{1}{2}}\sum_{i=1}^n \int
    w_i(\gamma^*) \mathbf{D}_i\ dN_i(t)-n^{\frac{1}{2}} \int \bar{\mathbf{d}}(t;\tau^w,\gamma^*)d\hat{F}_n(t; \gamma^*)\\
    &-n^{\frac{1}{2}} \int \{\frac{1}{s^{(0)}(t;\tau^w,\gamma^*)}[\mathbf{S}^{(1)}(t;\tau^w,\gamma^*)-\frac{\mathbf{s}^{(1)}(t;\tau^w,\gamma^*)}{s^{(0)}(t;\tau^w,\gamma^*)}S^{(0)}(t;\tau^w,\gamma^*)]\}d\hat{F}_n(t; \gamma^*)+o_p(1)\\
    =&n^{-\frac{1}{2}}\sum_{i=1}^n \int
    w_i(\gamma^*) (\mathbf{D}_i-\bar{\mathbf{d}}(t;\tau^w,\gamma^*)) dN_i(t)\\
    &-n^{\frac{1}{2}} \int \{\frac{1}{s^{(0)}(t;\tau^w,\gamma^*)}[\mathbf{S}^{(1)}(t;\tau^w,\gamma^*)-\frac{\mathbf{s}^{(1)}(t;\tau^w,\gamma^*)}{s^{(0)}(t;\tau^w,\gamma^*)}S^{(0)}(t;\tau^w,\gamma^*)]\}dF(t; \gamma^*)+o_p(1)\\
    =&n^{-\frac{1}{2}}\sum_{i=1}^n \int
    w_i(\gamma^*) (\mathbf{D}_i-\bar{\mathbf{d}}(t;\tau^w,\gamma^*)) dN_i(t)\\
    &-n^{\frac{1}{2}} \int \frac{1}{n}\sum_{i=1}^n w_i(\gamma^*) \frac{ Y_i(t)\exp\{\eta_i(\tau^w)\}}{s^{(0)}(t;\tau^w,\gamma^*)}[\mathbf{D}_i-\bar{\mathbf{d}}(t;\tau^w,\gamma^*)]dF(t; \gamma^*)+o_p(1)\\
    =&n^{-1/2}\sum_{i=1}^n \tilde{\Psi}_i(\tau^w,\gamma^*)+o_p(1),
\end{align*}}
where the influence function is
{\begin{align*}
\tilde{\Psi}_i(\tau^w, \gamma^*)
&=w_i(\gamma^*)\int \{\mathbf{D}_i-\bar{\mathbf d}(t;\tau^w, \gamma^*)\}\,dN_i(t)-
w_i(\gamma^*)\int
\frac{Y_i(t)e^{\eta_i(\tau^w)}}{s^{(0)}(t;\tau^w, \gamma^*)}
\{\mathbf{D}_i-\bar{\mathbf d}(t;\tau^w, \gamma^*)\}\,dF(t; \gamma^*).
\end{align*}}

Notice that $n^{-1/2}\sum_{i=1}^n \tilde{\Psi}_i(\tau^w,\gamma^*)$ and $n^{-1/2}\sum_{i=1}^n \psi_i(\tau^w,\gamma^*)$ are asymptotically equivalent, we can define the i.i.d.-equivalent contribution:
{\begin{equation*}
\tilde{\Phi}_i(\tau^w,\gamma^*)=
\begin{pmatrix}
\tilde{\Psi}_i(\tau^w,\gamma^*)\\
\pi_i(\gamma^*)
\end{pmatrix}.
\end{equation*}}

\subsection{Asymptotic normality and variance estimation}
Combining the influence contribution of the Cox score with the multinomial propensity-score estimating function
gives a stacked estimating equation. The resulting sandwich variance follows the standard
M-estimation formula for stacked estimating equations \citep{shu2021variance}.

By multivariate central limit theorem,

\begin{equation*}
    n^{-1/2}\sum_{i=1}^n \tilde{\Phi}_i(\tau^w,\gamma^*)
\xrightarrow{d}
\mathcal{N}\big(0,E[\tilde{\Phi}_i(\tau^w,\gamma^*)\tilde{\Phi}_i(\tau^w,\gamma^*)^{\prime}]\big).
\end{equation*}

A first-order Taylor expansion of \eqref{eq:joint} around $(\tau^w,\gamma^*)$ yields

\begin{align*}
\sqrt n
\begin{pmatrix}
\hat\tau^w-\tau^w\\
\hat\gamma-\gamma^*
\end{pmatrix}
=&
\mathbf{A}(\tau^w,\gamma^*)^{-1}
\frac{1}{\sqrt n}\sum_{i=1}^n
\phi_i(\tau^w,\gamma^*)
+o_p(1)\\
=&
\mathbf{A}(\tau^w,\gamma^*)^{-1}
\frac{1}{\sqrt n}\sum_{i=1}^n
\tilde{\Phi}_i(\tau^w,\gamma^*)
+o_p(1),
\end{align*}

where

\begin{equation*}
\mathbf{A}(\tau^w,\gamma^*) = -\frac{1}{n}\sum_{i=1}^n \frac{\partial}{\partial((\tau^{w})^\prime,(\gamma^*)^{\prime})} \phi_i(\tau^w,\gamma^*).
\end{equation*}

Thus

\begin{equation*}
    \sqrt n \begin{pmatrix}
\hat\tau^w-\tau^w\\
\hat\gamma-\gamma^*
\end{pmatrix}
\xrightarrow{d}
\mathcal{N}\!\left(
\mathbf 0,\ 
\left(- E\!\left[
\frac{\partial\ \phi_i(\tau^w,\gamma^*)}{\partial((\tau^w)^{\prime},(\gamma^*)^\prime)}
\right]\right)^{-1}
E[\tilde{\Phi}_i(\tau^w,\gamma^*)\tilde{\Phi}_i(\tau^w,\gamma^*)^{\prime}]
\left((- E\!\left[
\frac{\partial \ \phi_i(\tau^w,\gamma^*)}{\partial((\tau^w)^{\prime},(\gamma^*)^\prime)}
\right])^{-1}\right)^\prime
\right).
\end{equation*}

Replacing population quantities by their empirical counterparts yields
the sandwich variance estimator

\begin{equation} 
    \widehat{\mathrm{Var}}
\begin{pmatrix}
\hat{\tau}^w\\
\hat\gamma
\end{pmatrix}=\frac{1}{n}
\mathbf{A}(\hat\tau^w,\hat\gamma)^{-1} \,
\mathbf{B}(\hat\tau^w,\hat\gamma) \,
(\mathbf{A}(\hat\tau^w,\hat\gamma)^{-1})^\prime,
\end{equation}
where 
\begin{equation*}
\mathbf{A}(\tau^w,\gamma) = -\frac{1}{n}\sum_{i=1}^n \frac{\partial}{\partial((\tau^{w})^\prime,\gamma^{\prime})} \phi_i(\tau^w,\gamma), \qquad
\mathbf{B}(\tau^w,\gamma) = \frac{1}{n}\sum_{i=1}^n \Phi_i(\tau^w,\gamma)\Phi_i(\tau^w,\gamma)^{\prime},
\end{equation*}
and
\begin{equation*}
\Phi_i(\tau^w,\gamma)=
\begin{pmatrix}
\Psi_i(\tau^w,\gamma)\\
\pi_i(\gamma)
\end{pmatrix},
\end{equation*}
and
\begin{align*}
\Psi_i(\tau^w,\gamma) =
w_i(\gamma) \delta_i \left\{\mathbf{D}_i-\bar{\mathbf{D}}(U_i;\tau^w,\gamma) \right\} -
w_i(\gamma)
\sum_{j=1}^n \frac{w_j(\gamma)\delta_j \, \one(U_j\le U_i)\exp({\eta_i(\tau^w)})}
     {n S^{(0)}(U_j;\tau^w,\gamma)}
\left\{ \mathbf{D}_i-\bar{\mathbf{D}}(U_j;\tau^w,\gamma) \right\}.
\end{align*}

Then a consistent estimator of $\mathrm{Var}(\hat\tau^w)$ is the submatrix formed by the first $J$ rows and $J$ columns:

\begin{equation*}
    \widehat{\mathrm{Var}}(\hat\tau^w)=\frac{1}{n}
    \left[\mathbf{A}(\hat\tau^w,\hat\gamma)^{-1} \,
    \mathbf{B}(\hat\tau^w,\hat\gamma) \,
    (\mathbf{A}(\hat\tau^w,\hat\gamma)^{-1})^\prime
    \right]_{(1:J, 1:J)}.
\end{equation*}

\section{Illustration with data and R code for a two-way factorial randomized trial}

\subsection{The COMBINE study of Lin et al. (2016)}
This illustration concerns a randomized controlled trial with a two-way factorial design. Specifically, we applied the method to reanalyze the COMBINE study in Lin et al. \cite{lin2016simultaneous} to evaluate the efficacy of medication, behavioral therapy, and their combination for treating alcohol dependence. The Data are publicly available in the supplementary material of Lin et al. at \url{https://pmc.ncbi.nlm.nih.gov/articles/instance/5026867/bin/NIHMS760681-supplement-Supp_Info_Code.zip}. The clean data file \texttt{data-Lin.csv} that is used in the following code is available at \url{https://www2.stat.duke.edu/~fl35/data/data-Lin.csv}.

In the study, 1,226 eligible alcohol-dependent individuals were randomly assigned to receive medical treatment of naltrexone or placebo and were also randomly assigned to receive combined behavioral intervention (CBI).  The outcome is the time to the first heavy drinking day. Pre-treatment covariates are baseline percentage of days abstinent (within 30 days prior to the participant’s last drink) and research site. The no naltrexone no CBI, naltrexone alone, CBI alone, both naltrexone and CBI arms have 305, 302, 307, 312 participants, respectively, with a mean time to the first event of 21.8, 26.2, 24.6, and 23.0 weeks, respectively. 

We use a multinomial logistic model with a main effect of each covariate to estimate the propensity scores, and then fit the weighted Cox model (equation (3) in Section 2.2 of the main text) with a main effect of each of the four treatment arms, using the no naltrexone no CBI as the baseline group. Figure \ref{fig:lin-OW} shows the overlap weighted Kaplan-Meier curves of each arm. Because it is a randomized trial, there is little difference in the results between different weighting schemes, all similar to the unadjusted results. With overlap weights, naltrexone alone shows a significant protective effect (HR: 0.80; 95\% CI: 0.66--0.96), while CBI alone shows a smaller, nonsignificant effect (HR: 0.88; 95\% CI: 0.74--1.05). Similar to Lin et al. \cite{lin2016simultaneous}, their combination yields the smallest, nonsignificant effect (HR: 0.95; 95\% CI: 0.79--1.15), suggesting a negative interaction between naltrexone and behavioral interventions. These results are consistent with Figure \ref{fig:lin-OW}, where the naltrexone-alone arm shows the highest survival probability throughout follow-up.

\begin{figure}[ht]
    \centering
    \includegraphics[width=0.5\linewidth]{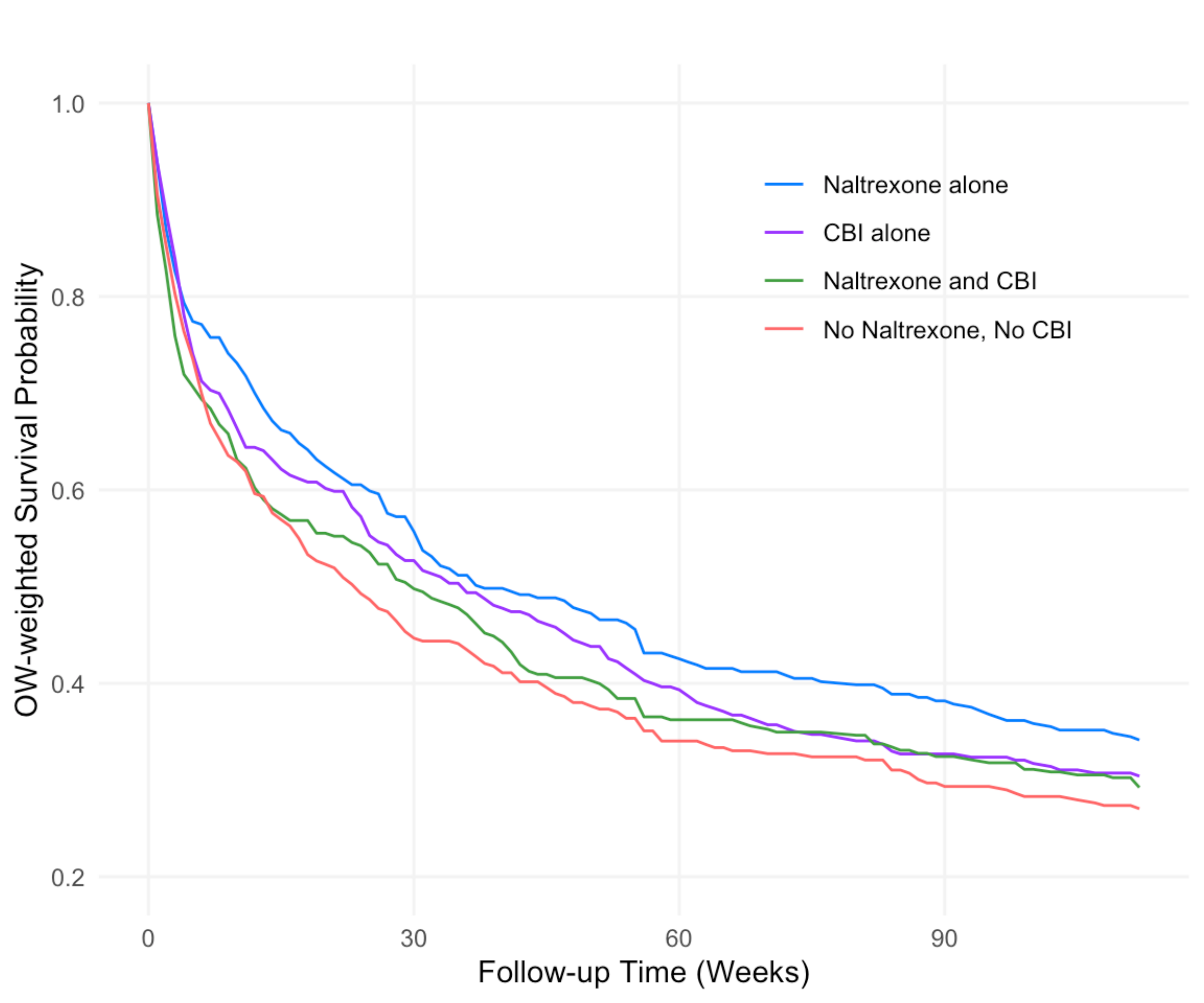}
    \caption{Overlap-Weighted Kaplan-Meier Curves of the COMBINE Study}
    \label{fig:lin-OW}
\end{figure}

\subsection{The R code}
\begin{lstlisting}[language=R]
# Load dependent packages
library(tidyverse)
library(PSsurvival)

# Import data (in .csv format). Data Description:
# Treatment: Two binary treatment variables: `NALTREXONE` and `THERAPY.`
# Outcome: `futime` (observed time); `relapse` (1 if event, 0 if censored, at `futime`)
# Covariates: `AGE`, `GENDER`, `T0_PDA`, `site`. Details see Lin et al. 2016.
data_lin2016 <- read.csv("data-Lin.csv", header = TRUE, sep = ",")

# Code 2*2 factorial treatments to a new variable `trt` with 4 unique groups.
data_lin2016 <- data_lin2016 %>% 
  mutate(trt = case_when(
    NALTREXONE == 0 & THERAPY == 0 ~ "No Naltrexone, No CBI",
    NALTREXONE == 1 & THERAPY == 0 ~ "Naltrexone alone",
    NALTREXONE == 0 & THERAPY == 1 ~ "CBI alone",
    NALTREXONE == 1 & THERAPY == 1 ~ "Naltrexone and CBI",
    TRUE ~ NA))

# Estimate (log) MHR by the OW-weighted Cox model using PSsurvival::marCoxph(),
# with "No Naltrexone, No CBI" being the reference group, and propensity scores
# estimated by multinomial logistic with covariates age, gender, T0_PDA, and site
marcox_lin_ow <- marCoxph(
  data = data_lin2016, 
  ps_formula = trt ~ AGE + GENDER + T0_PDA + site,
  time_var = "futime", 
  event_var = "relapse",
  reference_level = "No Naltrexone, No CBI",
  weight_method  = "OW", # if "IPW", use IPW weights
  variance_method = "bootstrap", # if "robust", use robust variance estimator
  B = 200, # Number of Bootstrap iterations 
  seed = 2026)

summary(marcox_lin_ow, round.digits = 2)

# OW-weighted Kaplan-Meier curves by treatment group using PSsurvival::weightedKM()
KM_lin_ow <- weightedKM(
  data = data_lin2016, 
  treatment_var = "trt",
  time_var = "futime",
  event_var = "relapse",
  ps_formula = trt ~ AGE + GENDER + T0_PDA + site,
  weight_method  = "OW" # if "IPW", use IPW weights
  )

## Plot 1: with default aesthetic settings
plot(KM_lin_ow, type = "Kaplan-Meier")

## Plot 2: with fine-tuned aesthetic settings
plot(KM_lin_ow, 
     type = "Kaplan-Meier", # if type="CR", draw cumulative risk curves
     include_CI = FALSE,
     strata_to_plot = c("Naltrexone alone", "CBI alone",
                        "Naltrexone and CBI", "No Naltrexone, No CBI"),
     strata_colors = c("#0080FF", "#9933FF", "#40A040", "#FF6666"),
     curve_width = 0.5,
     legend_title = "",
     plot_title = "",
     ylab = "OW-weighted Survival Probability",
     xlab = "Follow-up Time (Weeks)",
     ylim = c(0.2, 1.0)
) + ggplot2::theme(legend.position = "inside", legend.position.inside = c(0.75,0.8))


# Other summary statistics:
## 1. sample size per group
table(data_lin2016$trt)
## 2. Average time to the first event (relapse) per group
data_lin2016 %>%
  filter(relapse == 1) %>%
  group_by(trt) %>%
  summarise(ave_eventtime = mean(futime))
    
\end{lstlisting}

\section{Additional information of Simulations}
\begin{table}[ht]
\centering
\caption{Empirical event rates at selected time points (multi-treatment, $\psi=1$)}
\label{tab:eventrate_multi_psi1}
\begin{tabular}{lccccccc}
\toprule
Group & $t=0.2$ & $t=0.5$ & $t=1$ & $t=2$ & $t=3$ & $t=5$ & $t=10$ \\
\midrule
\multicolumn{8}{c}{Censoring rate: $25\%$} \\[0.3em]
Overall & 0.25 & 0.43 & 0.56 & 0.67 & 0.71 & 0.74 & 0.75 \\
$Z=0$   & 0.21 & 0.40 & 0.56 & 0.68 & 0.73 & 0.76 & 0.77 \\
$Z=1$   & 0.44 & 0.67 & 0.80 & 0.87 & 0.88 & 0.89 & 0.89 \\
$Z=2$   & 0.09 & 0.20 & 0.33 & 0.46 & 0.52 & 0.57 & 0.60 \\
\midrule
\multicolumn{8}{c}{Censoring rate: $50\%$} \\[0.3em]
Overall & 0.23 & 0.37 & 0.45 & 0.49 & 0.50 & 0.50 & 0.50 \\
$Z=0$   & 0.19 & 0.34 & 0.44 & 0.48 & 0.49 & 0.49 & 0.49 \\
$Z=1$   & 0.42 & 0.60 & 0.67 & 0.70 & 0.70 & 0.70 & 0.70 \\
$Z=2$   & 0.08 & 0.17 & 0.24 & 0.29 & 0.30 & 0.30 & 0.30 \\
\bottomrule
\end{tabular}
\end{table}

\begin{table}[ht]
\centering
\caption{Empirical event rates at selected time points (factorial design, $\psi=1$)}
\label{tab:eventrate_factorial_psi1}
\begin{tabular}{lccccccc}
\toprule
Group & $t=0.2$ & $t=0.5$ & $t=1$ & $t=2$ & $t=3$ & $t=5$ & $t=10$ \\
\midrule
\multicolumn{8}{c}{Censoring rate: $25\%$} \\[0.3em]
Overall & 0.25 & 0.43 & 0.57 & 0.67 & 0.71 & 0.74 & 0.75 \\
$Z=00$  & 0.20 & 0.39 & 0.54 & 0.66 & 0.71 & 0.74 & 0.75 \\
$Z=10$  & 0.44 & 0.66 & 0.79 & 0.86 & 0.88 & 0.88 & 0.88 \\
$Z=01$  & 0.08 & 0.19 & 0.31 & 0.44 & 0.50 & 0.55 & 0.57 \\
$Z=11$  & 0.26 & 0.47 & 0.62 & 0.73 & 0.77 & 0.79 & 0.80 \\
\midrule
\multicolumn{8}{c}{Censoring rate: $50\%$} \\[0.3em]
Overall & 0.23 & 0.37 & 0.45 & 0.49 & 0.50 & 0.50 & 0.50 \\
$Z=00$  & 0.19 & 0.33 & 0.42 & 0.47 & 0.48 & 0.48 & 0.48 \\
$Z=10$  & 0.41 & 0.59 & 0.66 & 0.69 & 0.69 & 0.69 & 0.69 \\
$Z=01$  & 0.08 & 0.16 & 0.23 & 0.28 & 0.29 & 0.29 & 0.29 \\
$Z=11$  & 0.24 & 0.40 & 0.49 & 0.53 & 0.54 & 0.54 & 0.54 \\
\bottomrule
\end{tabular}
\end{table}

\begin{figure}[ht]
    \centering
    \includegraphics[width=0.75\linewidth]{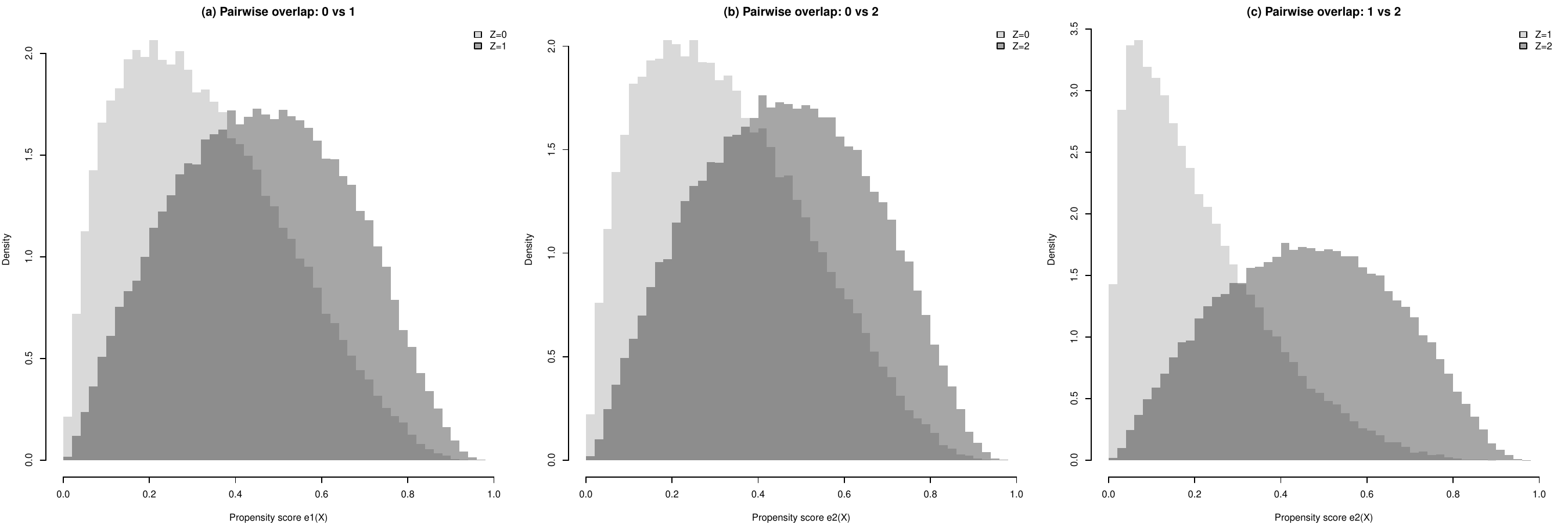}
    \caption{Multiple Treatment and Strong Overlap $\psi=1$}
    \label{fig:ps_overlap_psi1}
\end{figure}
\begin{figure}[ht]
    \centering
    \includegraphics[width=0.75\linewidth]{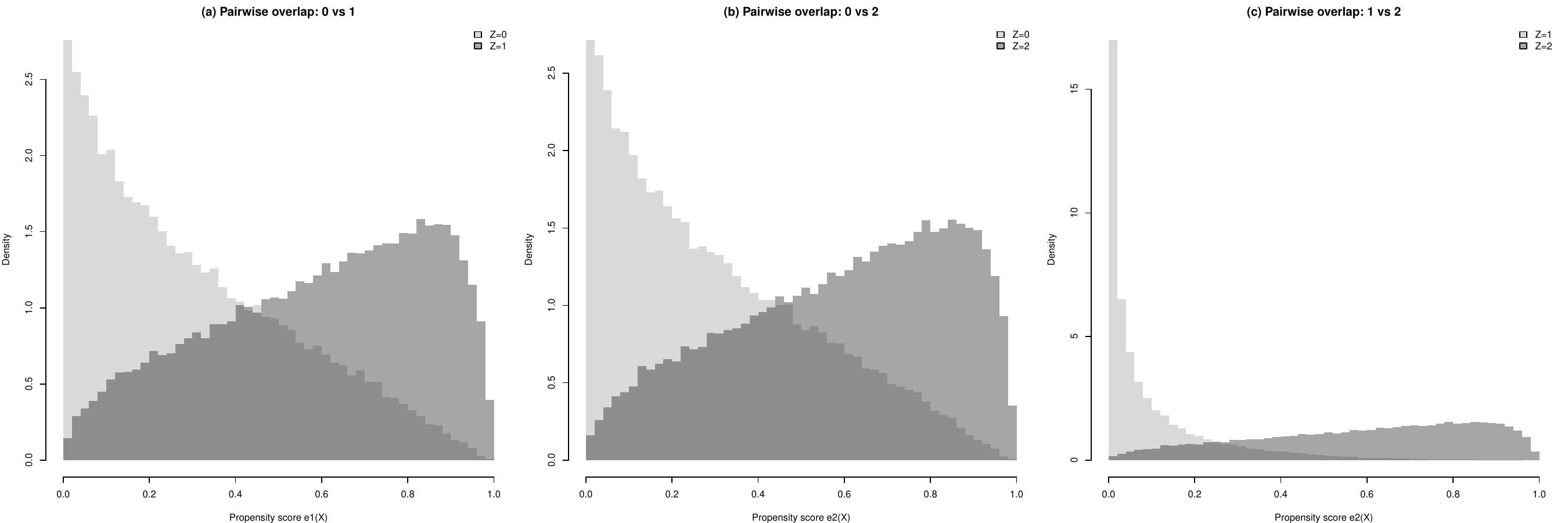}
    \caption{Multiple Treatment and Moderate Overlap $\psi=2$}
    \label{fig:ps_overlap_psi2}
\end{figure}
\begin{figure}[ht]
    \centering
    \includegraphics[width=0.75\linewidth]{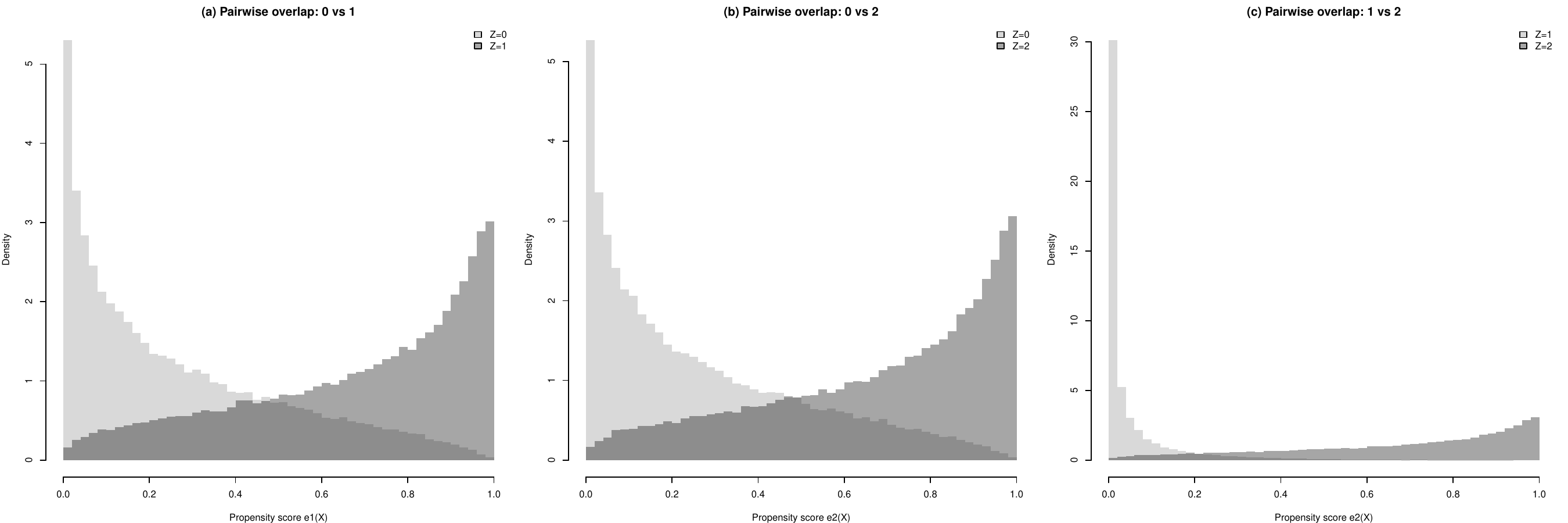}
    \caption{Multiple Treatment and Weak Overlap $\psi=3$}
    \label{fig:ps_overlap_psi3}
\end{figure}

\begin{figure}
    \centering
    \includegraphics[width=1\linewidth]{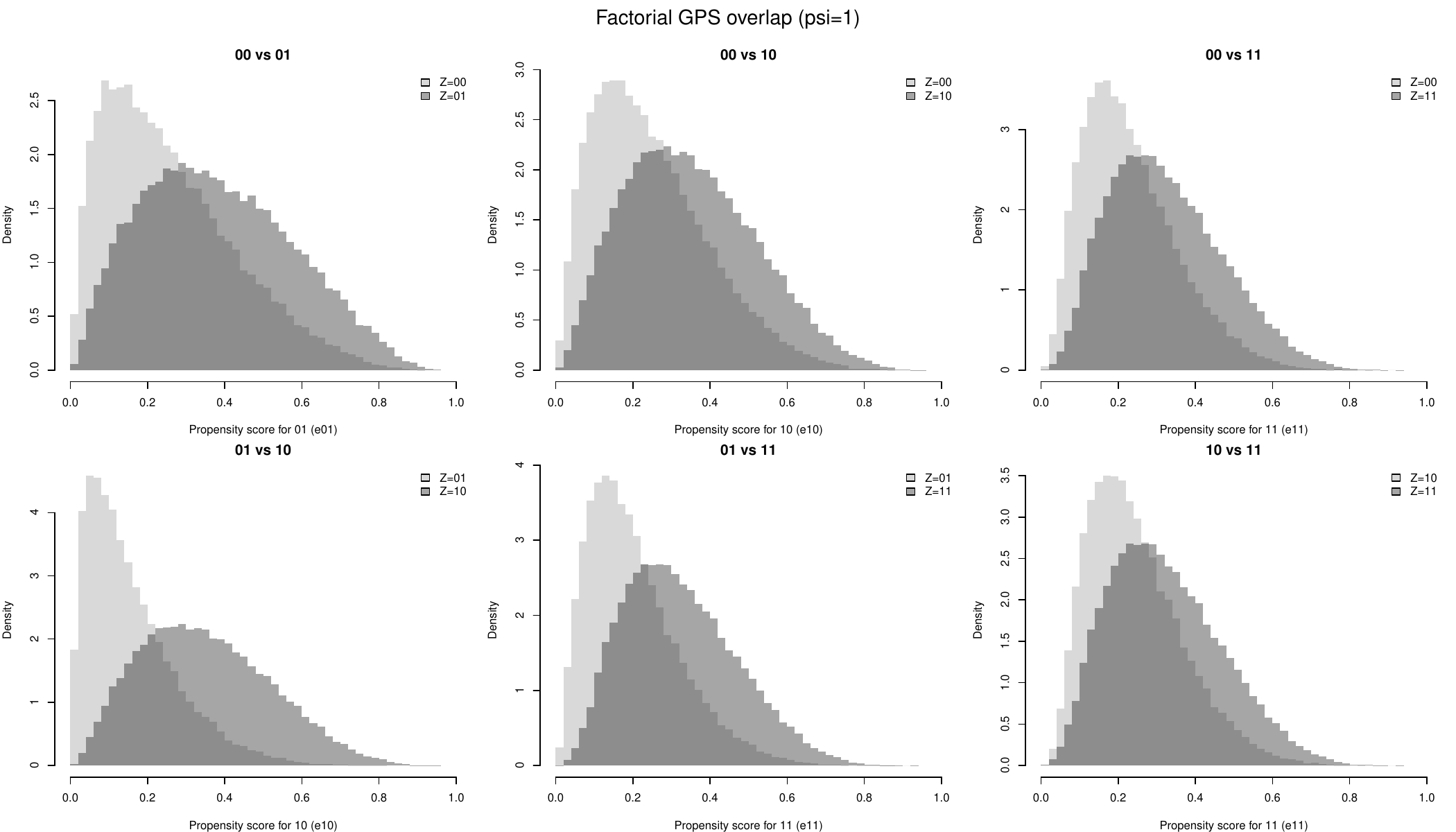}
    \caption{Two Way Factorial Design and Strong Overlap}
    \label{fig:ps_overlap_factorial_psi1}
\end{figure}
\begin{figure}
    \centering
    \includegraphics[width=1\linewidth]{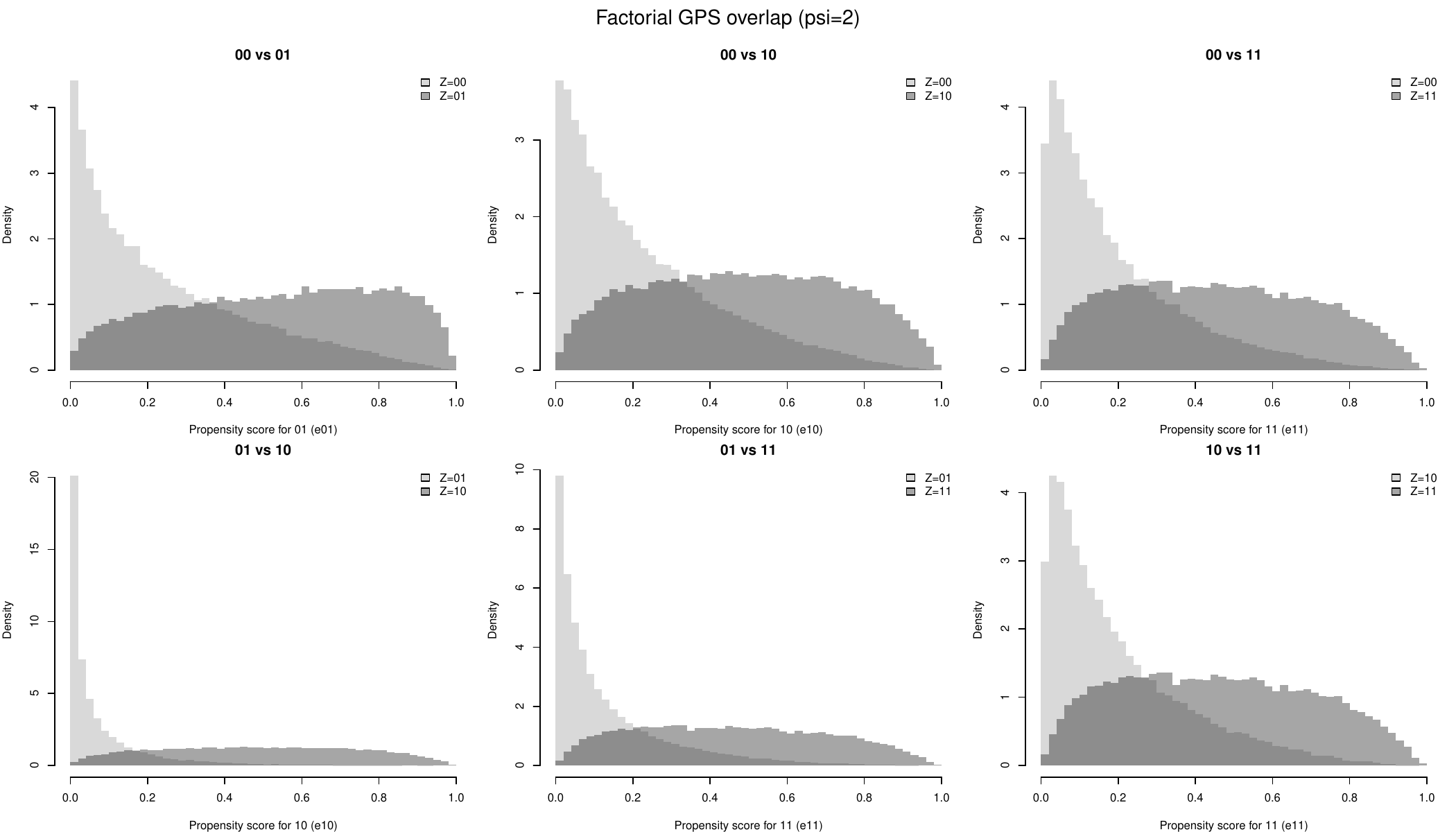}
    \caption{Two Way Factorial Design and Moderate Overlap}
    \label{fig:ps_overlap_factorial_psi2}
\end{figure}
\begin{figure}
    \centering
    \includegraphics[width=1\linewidth]{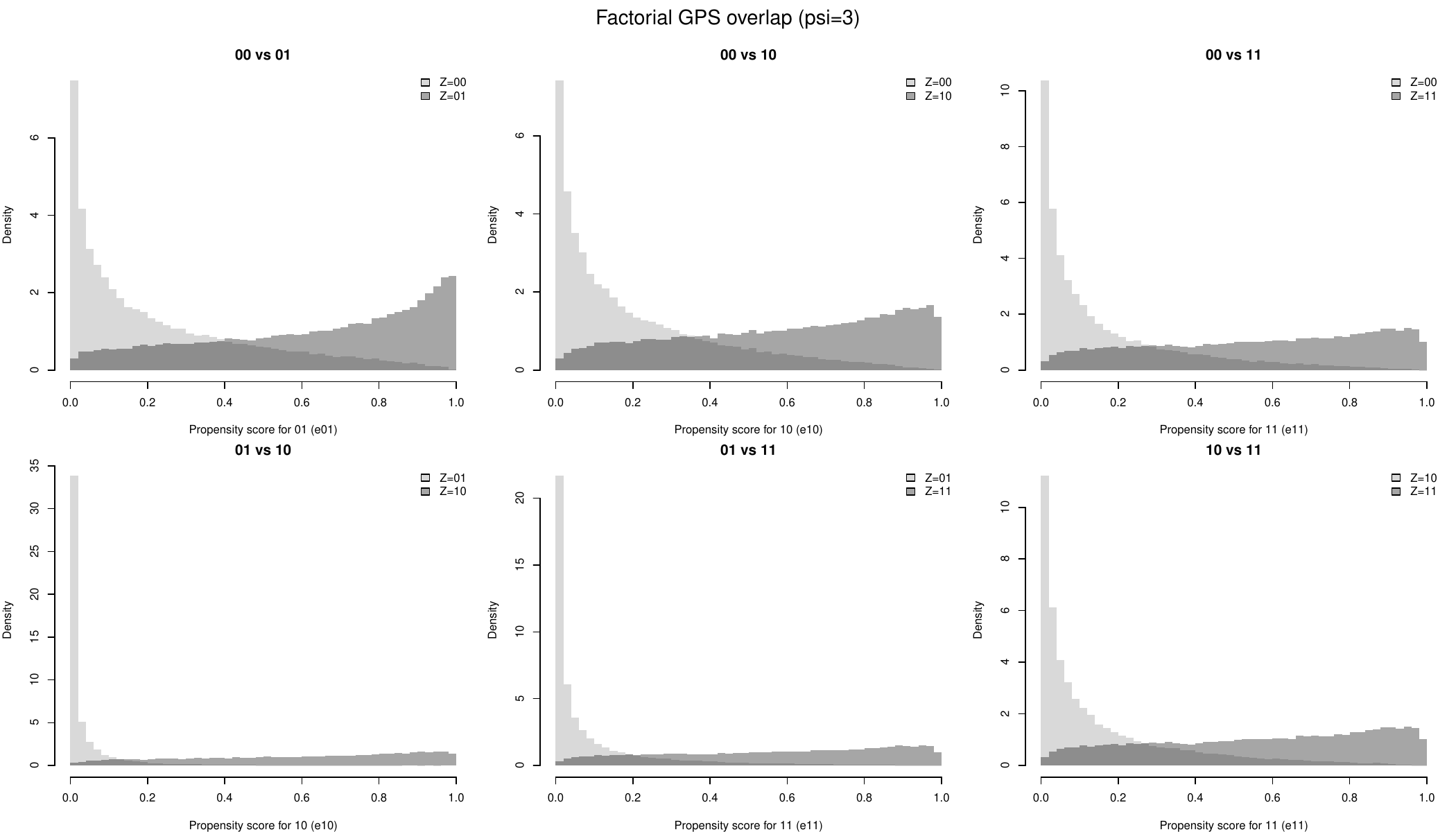}
    \caption{Two Way Factorial Design and Weak Overlap}
    \label{fig:ps_overlap_factorial_psi3}
\end{figure}

\end{document}